\documentclass[11pt]{article}

\usepackage[margin=1in]{geometry}

\usepackage[utf8]{inputenc}
\usepackage[T1]{fontenc}

\usepackage{amsmath,amssymb,amsfonts,amsthm}
\usepackage{mathtools}
\usepackage{bbm}

\usepackage{graphicx}
\usepackage{booktabs}
\usepackage{nicefrac}
\usepackage{algorithm}
\usepackage{algorithmic}
\usepackage{float}
\usepackage{nicematrix}
\usepackage{subfigure}

\usepackage[table]{xcolor}
\usepackage{url}
\usepackage{natbib}
\usepackage{hyperref}
\usepackage[capitalize,noabbrev]{cleveref}

\usepackage{microtype}

\usepackage{tikz}
\usetikzlibrary{arrows.meta,positioning,fit}

\usepackage{comment}
\usepackage{lastpage}

\newtheorem{theorem}{Theorem}
\newtheorem{lemma}[theorem]{Lemma}
\newtheorem{proposition}[theorem]{Proposition}
\newtheorem{corollary}[theorem]{Corollary}

\theoremstyle{definition}

\theoremstyle{remark}

\DeclareMathOperator{\cov}{cov}

\title{Smoothed Score Queries and the Complexity of Sampling}

\author{
  Jingbo Liu\thanks{Department of Statistics, University of Illinois Urbana--Champaign. Email: \texttt{jingbol@illinois.edu}. This research 
was supported in part by NSF Grant DMS-2515510. The author declares no competing interests.}
}

\hypersetup{hidelinks,pdftitle={Smoothed Score Queries and the Complexity of Sampling},pdfauthor={Jingbo Liu}}
\date{}

\begin{document}

\maketitle

\begin{abstract}
We study the query complexity of sampling from high-dimensional Gaussian distributions using gradient information.  In the standard oracle model, exact gradients expose only matrix-vector products with the precision matrix, leading to polynomial approximation barriers and a characteristic \(\sqrt{\kappa}\) dependence on the known spectral range.  We show that this barrier disappears when the sampler is allowed to query \emph{smoothed scores}, namely gradients of the logarithms of the Gaussian-convolved densities.  For a centered Gaussian target with precision matrix \(\Lambda\), a smoothed-score query at noise level \(\tau\) gives access to the resolvent \((\Lambda+\tau^{-1}I)^{-1}\).  Combining geometrically spaced noise levels with sinc-quadrature rational approximation, we obtain a sampler with
\[
    q
    =
    O\!\left(
        \bigl(\log\kappa+\log(e\sqrt d/\delta_{\rm TV})\bigr)
        \log(e\sqrt d/\delta_{\rm TV})
    \right)
\]
smoothed-score queries for total variation error \(\delta_{\rm TV}\), improving the condition-number dependence from \(\sqrt{\kappa}\) to logarithmic.
We also study finite-bit gradient oracles.  Using coordinatewise quantization of the transformed smoothed-score answers and a final dithering step, we obtain a sampling scheme whose total communicated gradient information is polylogarithmic in \(\kappa\); in particular, for fixed dimension and accuracy, the bit complexity is \(O(\log^2\kappa)\).  To complement these upper bounds, we introduce a channel-synthesis, or reverse-Shannon, converse technique for sampling lower bounds.  This converts total-variation simulation guarantees into communication requirements and yields an \(\Omega(d\log\kappa)\) bit lower bound for \(d\ge2\) and fixed total variation error below one, as \(\kappa\to\infty\).  Together, these results identify smoothed scores as a provably more informative oracle for sampling and give nearly matching upper and lower bounds for its finite-bit complexity.
\end{abstract}

\section{Introduction}

Annealing is a standard way to make high-dimensional sampling easier.  One common approach is \emph{tilting}: given a target potential $U$ (negative log of the target distribution) and a known Gaussian potential $U_G$, one constructs an interpolating family $U_t=(1-t)U+tU_G$ and samples along this path.  Another approach is \emph{smoothing}: one convolves the target distribution with a Gaussian kernel and queries the score of the smoothed law,
which arises naturally, for example, in training diffusion models.
These two procedures are often used for similar algorithmic purposes, but they expose fundamentally different information to a sampler: the score function of the tilted distribution, $-\nabla U_t=-(1-t)\nabla U-t\nabla U_G$,
contains the same information as $\nabla U$ for $t<1$, whereas the score of the smoothed law may provide more.  This paper studies the distinction at the level of oracle and communication complexity.

We focus on Gaussian targets $\pi=\mathcal N(0,\Sigma)$ with precision matrix $\Lambda=\Sigma^{-1}$ in a known spectral interval, rescaled to $[1,\kappa]$. Here $\kappa$ bounds the condition number and also describes uncertainty about the scale. Gaussian sampling is a subroutine in spatial statistics, Gaussian-process inference, and Bayesian inverse problems \citep{aune2013iterative,chow2014preconditioned}. In Laplace-based uncertainty quantification, once a posterior mode $x_\star$ has been found and its quadratic approximation judged adequate, the remaining subproblem is sampling from $\mathcal N(x_\star,H^{-1})$. Locating the mode and certifying that approximation are separate tasks.

In the standard first-order oracle model, querying the target score gives matrix-vector products with $\Lambda$.  Sampling from $\mathcal N(0,\Lambda^{-1})$ requires applying $\Lambda^{-1/2}$ to a standard Gaussian vector, and Krylov and Chebyshev methods approximate $x^{-1/2}$ on $[1,\kappa]$ by polynomials.  This is the source of the familiar $\sqrt\kappa$ dependence appearing in query lower and upper bounds for Gaussian sampling, including the Gaussian result of \citet{chewi2023query}.  By contrast, a smoothed-score query at noise level $\tau$ gives
$$
    s_\tau(y)=\nabla\log(\pi*\mathcal N(0,\tau I))(y)=-(\Sigma+\tau I)^{-1}y,
$$
which is equivalent to applying the resolvent $(\Lambda+\tau^{-1}I)^{-1}$.  Thus smoothing changes the approximation problem from polynomial approximation to rational approximation
by functions of the form $\frac1{x+\tau^{-1}}$, giving rise to an improved $\log\kappa$ dependence.  The improvement in this paper comes from exploiting this rational structure. Smoothing is used only in the queries: the output approximates the original, unsmoothed target $\pi$. Total variation then contracts under any downstream Markov kernel or measurable post-processing; it does not by itself control arbitrary unbounded statistics or repeated-sample error.

\paragraph{Oracle models.}
We consider two smoothed-score oracle models. 
\begin{itemize}
\item In the \emph{exact model}, an algorithm chooses a noise level $\tau>0$ and query point $y\in\mathbb R^d$ and receives the full vector $s_\tau(y)$.
\item In the \emph{finite-bit model}, the algorithm receives only a transcript of encoded smoothed-score information; the total number of transmitted bits is denoted by $Q$.  This model describes quantization of exact responses, rather than general score-estimation error. Each full-vector response with $B$ bits per coordinate costs $dB$ bits.
\end{itemize}  

\paragraph{Informal theorem.}
For centered Gaussian targets with known precision spectrum in $[1,\kappa]$, exact smoothed-score queries allow total-variation sampling with
$$
    q=O\left((\log\kappa+\log(e\sqrt d/\delta_{\mathrm{TV}}))\log(e\sqrt d/\delta_{\mathrm{TV}})\right)
$$
queries, where the implicit constant in the $O$ notation is universal (see Theorems~\ref{thm:exact-upper-short} and~\ref{thm:exact-upper-independent}).  A coordinatewise quantized version with final dithering uses total communication $Q=\operatorname{polylog}(\kappa)$ for fixed dimension and accuracy (Theorem~\ref{thm:coordinate-quantized-upper}). The same exact-query complexity extends to uncentered Gaussians with constant overhead (Proposition~\ref{prop:noncentered-reduction}). On the converse side, 
we introduce a method based on channel synthesis \citep{bennett2014quantum,cuff2013distributed},
which yields a finite-transcript lower bound of order $d\log\kappa$ for fixed TV error below one, $d\ge2$, and sufficiently large $\kappa$ (Theorem~\ref{thm:gaussian-bit-lower-fixed-delta}).

\paragraph{Contributions.}
Our results have two directions.
\begin{itemize}
    \item \textbf{Upper bounds from smoothed scores.}
    We show that smoothed-score queries reveal resolvents of the precision matrix.  Combining geometrically spaced noise levels with sinc-quadrature rational approximation gives a sampler with logarithmic dependence on $\kappa$, improving over the $\sqrt\kappa$ behavior of unsmoothed score/Krylov methods.  This gives a theoretical basis for why smoothing can be more powerful than tilting for annealing: smoothing changes the oracle from polynomial access to rational access.  We also give a coordinatewise quantization scheme (Theorem~\ref{thm:coordinate-quantized-upper}) and an uncentered Gaussian reduction (Proposition~\ref{prop:noncentered-reduction}).

    \item \textbf{Lower bounds from channel synthesis.}
    We introduce a channel-synthesis, or reverse-Shannon, method for lower bounding sampling communication.  A finite-bit sampler for a class $\{p_x:x\in\mathcal X\}$ is a simulator for the channel $x\mapsto Y\sim p_x$.  Channel coding then gives a converse for simulation.  
    Compared with Fano-style arguments for query lower bounds \citep{chen2023sampling,chewi2022fisher,chewi2023query}, our approach applies directly to sampling without requiring reduction to parameter estimation,
    and also gives stronger lower bound when the allowed TV error is close to one (see Theorem~\ref{thm:fixed-kappa-large-d}).
\end{itemize}

\paragraph{Limitations.}
Theorems~\ref{thm:exact-upper-short}--\ref{thm:coordinate-quantized-upper} use exact Gaussian structure: score differences provide resolvent access to one fixed precision matrix. Non-Gaussian scores are generally nonlinear, so neither local Gaussianity nor strong log-concavity alone justifies this construction. Theorem~\ref{thm:coordinate-quantized-upper} controls clipping and quantization of exact scores; it does not establish robustness to general learned-score error. The converse in Theorem~\ref{thm:channel-synthesis-converse} concerns finite transcripts and does not apply to exact real-valued responses, which can reveal infinitely many bits. Extensions beyond Gaussians, general score-error guarantees, and exact-query converses remain open.

\paragraph{Related work.}

\emph{Upper bounds.}
Ordinary Gaussian scores give matrix-vector access to $\Lambda$, so Krylov and Lanczos samplers approximate $\Lambda^{-1/2}Z$ by polynomials \citep{aune2013iterative,chow2014preconditioned,musco2018stability,chen2022lanczos,chewi2023query}. Rational inverse-square-root approximation is also classical: \citet{aune2013iterative} and \citet{pleiss2020fast} use quadrature and shifted linear solves. In particular, the contour-quadrature method of \citet{pleiss2020fast} uses logarithmically many shifts in $\kappa$, but its multi-shift MINRES solves share a Krylov space and retain a $\sqrt\kappa$-type matrix-vector cost without effective preconditioning. We analyze oracle complexity under the direct resolvent access supplied by smoothed scores, establishing finite-bit upper bounds and a channel-synthesis converse.

Noise-indexed scores arise in denoising and diffusion models \citep{song2019generative,ho2020denoising,song2021scorebased}, in contrast to potential interpolation in annealing \citep{neal2001annealed}. General-target guarantees with estimated scores require additional assumptions \citep{chen2023sampling}; dimension and high-accuracy refinements include \citet{chen2023probability,jiao2024instance,zhang2025sublinear,chen2026highaccuracy}, and higher-order Langevin methods include \citet{mahajan2025picard}. Our oracle models multi-noise score information through exact evaluations on a finite, predetermined geometric grid.

For strongly log-concave targets, \citet{wainwright2026fast} gives at most $2+\log_2\kappa$ well-conditioned stages, with total score complexity bounded by
$(2+\log_2\kappa)N_{\rm SLC}(\varepsilon/(2+\log_2\kappa))$, where $N_{\rm SLC}$ is the cost of a well-conditioned inner sampler. His displayed high-accuracy KL implementation scales as $\sqrt d(1+\log\kappa)\log((1+\log\kappa)/\varepsilon)$. Exact Gaussian linearity removes that inner sampler here, giving polylogarithmic dependence on $d$ and $O(\log\kappa)$ queries for fixed $d,\varepsilon$. Pinsker's inequality converts KL accuracy to TV accuracy at the cost of a constant factor inside this logarithm. Theorems~\ref{thm:exact-upper-short} and~\ref{thm:exact-upper-independent} complement his results for general strongly log-concave targets and score-error stability. We also quantify finite-bit communication and prove a corresponding converse.

\emph{Lower bounds.}
Query lower bounds for sampling have been developed in several oracle models,
including lower bounds for strongly log-concave and smooth targets,
Fisher-information lower bounds for broader distribution classes, and oracle lower
bounds for Gaussian sampling
\citep{chewi2022query,chewi2022fisher,chewi2023query,lu2023upper}.  In the
ordinary-gradient Gaussian model with first-order oracle, the Krylov and polynomial-approximation viewpoint
leads to \(\sqrt{\kappa}\)-type barriers, again up to logarithmic factors. 
This does not rule out the possibility of faster sampling with smoothed-score oracles.

The first information-theoretic query lower bound for smoothed scores is due to \citet{xun2026query},
who study sampling with access to learned smoothed-score estimates.  They show
that, for broad \(d\)-dimensional target classes and polynomially accurate score
estimates, any sampler requires \(\widetilde{\Omega}(\sqrt d)\) adaptive score
queries.  Their result captures a dimension-dependent difficulty arising from the
need to search across many noise levels in an approximate-score oracle model.
Related works also establish \(\widetilde{\Omega}(\sqrt d)\)-type barriers for
specialized diffusion samplers \citep{gao2025convergence,JZL25}; these lower
bounds are not information-theoretic, but address  algorithm-specific diffusion settings.

Our lower-bound argument follows a different route, inspired by channel synthesis
and reverse Shannon theorems
\citep{cuff2013distributed,bennett2014quantum,berta2011quantum}.  Instead of
reducing sampling to recovery of a hidden parameter, as in Fano-style arguments in prior works, we
view a distribution class as a channel \(\theta \mapsto P_\theta\) and lower bound
the finite communication required to synthesize one output sample from this
channel.  This yields a bit-complexity converse rather than an estimation or
testing lower bound.

\paragraph{Organization.}
Section~\ref{subsec:gaussian-sampling-setup} formulates the problem and the two query models.
Section~\ref{sec:upper-bounds} gives the centered, uncentered, and finite-bit Gaussian upper bounds. Section~\ref{sec:lb-channel-synthesis} develops a general channel-synthesis converse for finite-bit sampling. Section~\ref{sec:capacity-calculation} proves a channel coding achievability bound for condition-number-bounded Gaussian families and combines it with the channel-synthesis converse to obtain query lower bounds. Technical proofs and numerical validation are deferred to the appendices.

\section{Gaussian sampling setup}
\label{subsec:gaussian-sampling-setup}

\paragraph{Notation.}
Throughout the paper, \(\log\) denotes the natural logarithm unless a base is explicitly indicated.  We write \(\log_2\) for the base-two logarithm, which is used for bit lengths, coding rates, capacities, and transcript sizes.  
Asymptotic notation such as \(O(\cdot)\) hides universal constants.
For a random variable \(Y\), \(\mathcal L(Y)\) denotes its law.  We write \(d_{\mathrm{TV}}(P,Q)\) for total variation distance.  For symmetric matrices, \(A\preceq B\) denotes the Loewner order, and \(\operatorname{spec}(A)\) denotes the spectrum of \(A\).  The identity matrix in dimension \(d\) is denoted by \(I_d\), or simply \(I\) when the dimension is clear.  Constants denoted by \(C,c>0\) are universal unless explicitly indexed, for example \(C_\rho\) may depend on \(\rho\).

\paragraph{Gaussian targets.}
We consider centered Gaussian targets
\begin{align}
    \pi=\mathcal N(0,\Sigma),
    \qquad
    \Lambda=\Sigma^{-1},
    \qquad
    \operatorname{spec}(\Lambda)\subseteq[1,\kappa].
    \label{e_target}
\end{align}
The dimension is $d$. We assume a known precision interval $[m,L]$; rescaling the sample by $\sqrt m$ gives \eqref{e_target} with $\kappa=L/m$ and $\Sigma\preceq I$. Thus $\kappa$ is a spectral-range parameter and an upper bound on the condition number, not necessarily its value (for example, $\Lambda=\kappa I$ has condition number one). Knowledge of the scale is an assumption. In the exact-query model, an unknown Gaussian center can be handled with at most $2q+2$ queries in place of a centered sampler's $q$ queries; see Proposition~\ref{prop:noncentered-reduction}. Theorem~\ref{thm:coordinate-quantized-upper} also gives a finite-bit extension when the mean is known or its norm has a known bound.

For $\tau>0$, let
$$
    s_\tau(y)
    =
    \nabla\log\left(\pi*\mathcal N(0,\tau I)\right)(y)
$$
be the score of the Gaussian-smoothed target. Since $\pi$ is Gaussian,
$$
    \pi*\mathcal N(0,\tau I)
    =
    \mathcal N(0,\Sigma+\tau I),
    \qquad
    s_\tau(y)
    =
    -(\Sigma+\tau I)^{-1}y.
$$
Thus a smoothed-score query gives information about a shifted covariance inverse. The algorithms below use two oracle models.

\paragraph{Exact smoothed-score oracle.}
In the exact-query model, an algorithm may adaptively choose a noise level $\tau_t>0$ and query point $y_t\in\mathbb R^d$. The oracle returns the full vector
$$
    g_t
    =
    s_{\tau_t}(y_t)
    =
    \nabla\log\left(\pi*\mathcal N(0,\tau_t I)\right)(y_t).
$$
The query point $y_t$ and noise level $\tau_t$ may depend on the algorithm's internal randomness and on all previous oracle responses. The query complexity is the number of oracle calls, denoted by $q$.

\paragraph{Finite-bit smoothed-score oracle.}
In the finite-bit model, the oracle does not return the full vector
$g_t=s_{\tau_t}(y_t)$.  The query index $t$ ranges over the integers
$t=1,\dots,q$, where $q$ is the total number of oracle calls.  At query
$t$, the algorithm adaptively chooses a noise level $\tau_t>0$, a query
point $y_t\in\mathbb R^d$, a bit budget $B_t\in\mathbb N$, and an encoder
$$
    \mathsf{Enc}_t:\mathbb R^d\to\{0,1\}^{B_t}.
$$
A remote provider computes the smoothed score $g_t=s_{\tau_t}(y_t)$; an encoder, selected by the sampler, sends only the message
$$
    b_t=\mathsf{Enc}_t(g_t)\in\{0,1\}^{B_t}.
$$
The algorithm may then apply an arbitrary reconstruction or decision rule to the transcript.  Equivalently, it may choose a decoder
$$
    \mathsf{Dec}_t:\{0,1\}^{B_t}\to\mathbb R^d
$$
and form a reconstructed vector $\widehat g_t=\mathsf{Dec}_t(b_t)$, but the model does not impose a fixed reconstruction loss.  The only communication constraint is the total number of transmitted bits
$$
    Q=\sum_{t=1}^{q}B_t.
$$
The encoders, decoders, query points, noise levels, and bit allocations may all be chosen adaptively from previous messages and the algorithm's internal randomness. The sampler has no other target-dependent information. The provider and encoder see the exact score; the sampler sees only the message. Query descriptions are not charged. For Theorem~\ref{thm:channel-synthesis-converse}, $Q$ is a uniform worst-case bound on the total response length, rather than an expected length.

For the centered sampler in Theorem~\ref{thm:coordinate-quantized-upper}, we use a simple nonadaptive allocation: each query uses a straightforward quantization encoder with the same coordinatewise bit depth $B$, so the total number of bits per query is $dB$ and $Q=dBq$.  The lower bounds in Theorems~\ref{thm:gaussian-bit-lower-fixed-delta} and~\ref{thm:fixed-kappa-large-d} allow arbitrary adaptive choices of $B_t$ in the model above, including nonuniform bit allocations across queries.

\section{Upper bounds}
\label{sec:upper-bounds}

Sampling reduces to a matrix function: $\Lambda^{-1/2}Z\sim\mathcal N(0,\Lambda^{-1})$ for $Z\sim\mathcal N(0,I_d)$. Gaussian scores give
\[
    \tau Z+\tau^2s_\tau(Z)=(\Lambda+\tau^{-1}I)^{-1}Z,
    \qquad
    x^{-1/2}=\frac1\pi\int_0^\infty\frac{t^{-1/2}}{x+t}\,dt.
\]
Discretizing this integral on geometric scales gives a rational function $r(x)=\sum_j c_j/(x+\alpha_j)$, implemented by weighted queries. Relative spectral approximation controls the covariance of $r(\Lambda)Z$; Gaussian KL and Pinsker's inequality then give the TV guarantee. Detailed quadrature parameters appear in Appendix~\ref{subsec:sinc-grid}.

For a scalar precision $\lambda\in[1,\kappa]$, the same construction is simply $X_j=(\lambda+\alpha_j)^{-1}Z$ and $Y=r(\lambda)Z\approx\lambda^{-1/2}Z$. Reusing $Z$ applies one matrix function to one Gaussian vector. Independent query locations also work, with a different covariance normalization as in Algorithm~\ref{alg:exact-rational-sampler-independent}. Algorithm~\ref{alg:coordinate-quantized-rational-sampler} quantizes the weighted responses and adds an independent Gaussian dither to obtain a continuous output law.

\subsection{Exact smoothed-score queries}
\label{subsec:exact-smoothed-score}

The exact sampler is Algorithm~\ref{alg:exact-rational-sampler},
which is built on 
Lemma~\ref{lem:sinc-quadrature},
and the corresponding query complexity is $q=|\mathcal J|=M+N+1$.
Our result is as follows.

\begin{algorithm}[t]
\caption{Exact rational sampler from smoothed-score queries}
\label{alg:exact-rational-sampler}
Input: $\delta_{\rm TV}\in(0,1)$, $\kappa\ge 1$.
\begin{enumerate}
\item Set
$\eta
    =
    \frac{\delta_{\mathrm{TV}}}{4\sqrt d}$,
and compute $h$, $M$, $N$, $\mathcal{J}$, $\alpha_j$, $c_j$ as in Section~\ref{subsec:sinc-grid}.
Furthermore, set $\tau_j:=\alpha_j^{-1}$.
    \item Draw $Z\sim\mathcal N(0,I_d)$.
    \item For each $j\in\mathcal J$, query $s_{\tau_j}(Z)$ and form
    $X_j=\tau_j Z+\tau_j^2s_{\tau_j}(Z)$.
    \item Output
    $$
        Y=\sum_{j\in\mathcal J}c_jX_j.
    $$
\end{enumerate}
\end{algorithm}

\begin{theorem}[Exact smoothed-score upper bound]
\label{thm:exact-upper-short}
Consider arbitrary $\kappa\ge 1$, $\delta_{\rm TV}\in (0,1)$, 
centered Gaussian target satisfying \eqref{e_target},
and assume the exact smoothed-score oracle.
Run Algorithm~\ref{alg:exact-rational-sampler}.
Then
\[
    d_{\mathrm{TV}}\left(\mathcal L(Y),\pi\right)
    \leq
    \delta_{\mathrm{TV}},
\]
and the number of exact smoothed-score queries is
\[
    q
    =
    O\left(
        \left(
            \log\kappa+\log(e\sqrt d/\delta_{\mathrm{TV}})
        \right)
        \log(e\sqrt d/\delta_{\mathrm{TV}})
    \right),
\]
where $O(\cdot)$ hides universal constants.

For $\pi=\mathcal N(\mu,\Sigma)$ with arbitrary unknown $\mu$ and the same precision bounds, Proposition~\ref{prop:noncentered-reduction} gives the same TV guarantee and asymptotic query bound.
\end{theorem}
\begin{proof}[Proof sketch]
The resolvent identity gives $Y=r(\Lambda)Z$. The first estimate in Lemma~\ref{lem:sinc-quadrature} bounds each covariance ratio $\lambda r(\lambda)^2$ near one. Gaussian KL and Pinsker give the TV bound, while the grid size gives $q$. See Appendix~\ref{app:proof-exact-upper} for details.
\end{proof}

In contrast to the $\tilde{O}(\sqrt{\kappa})$ dependence arising from polynomial approximation under ordinary gradient access, 
Theorem~\ref{thm:exact-upper-short} shows that the smoothed-score oracle achieves only logarithmic dependence on $\kappa$.

\begin{algorithm}[t]
\caption{Exact rational sampler with independent smoothed-score queries}
\label{alg:exact-rational-sampler-independent}
Input: $\delta_{\rm TV}\in(0,1)$, $\kappa\ge 1$.
\begin{enumerate}
    \item Set
\[
    \eta
    =
    \frac{\delta_{\rm TV}}
    {8\sqrt d\,\log(C_0\sqrt d/\delta_{\rm TV})},
\]
    where $C_0$ was defined in \eqref{eq:C0-choice}. Compute
    $h$, $M$, $N$, $\mathcal J$, $\alpha_j$, $c_j$, and $L_h$
    as in Section~\ref{subsec:sinc-grid}. Set $\tau_j:=\alpha_j^{-1}$.

    \item For each $j\in\mathcal J$, independently draw
    \[
        Z_j\sim \mathcal N(0,I_d).
    \]

    \item For each $j\in\mathcal J$, query $s_{\tau_j}(Z_j)$ and form
    \[
        X_j
        =
        \tau_j Z_j+\tau_j^2s_{\tau_j}(Z_j).
    \]

    \item Output
    \[
        Y
        =
        \frac{1}{\sqrt{L_h}}
        \sum_{j\in\mathcal J}c_jX_j.
    \]
\end{enumerate}
\end{algorithm}

While Algorithm~\ref{alg:exact-rational-sampler} uses only smoothed score queries at one point $Z$,
an alternative approach in Algorithm~\ref{alg:exact-rational-sampler-independent} uses queries at independently sampled locations.
This variant leads to a covariance representation governed by the second approximation in Lemma~\ref{lem:sinc-quadrature}.
\begin{theorem}[Exact smoothed-score upper bound with independent queries]
\label{thm:exact-upper-independent}
Consider arbitrary $\kappa\ge 1$, $\delta_{\rm TV}\in(0,1)$,
a centered Gaussian target satisfying \eqref{e_target},
and assume the exact smoothed-score oracle.
Run Algorithm~\ref{alg:exact-rational-sampler-independent}.
Then
\[
    d_{\mathrm{TV}}\left(\mathcal L(Y),\pi\right)
    \leq
    \delta_{\rm TV}.
\]
Moreover, the number of exact smoothed-score queries satisfies
\[
    q
    =
    O\left(
        \left(
            \log\kappa+\log(e\sqrt d/\delta_{\rm TV})
        \right)
        \log(e\sqrt d/\delta_{\rm TV})
    \right),
\]
where $O(\cdot)$ hides universal constants.

For $\pi=\mathcal N(\mu,\Sigma)$ with arbitrary unknown $\mu$ and the same precision bounds, Proposition~\ref{prop:noncentered-reduction} gives the same TV guarantee and asymptotic query bound.
\end{theorem}

\begin{proof}[Proof sketch]
Independence gives covariance $L_h^{-1}\sum_j c_j^2(\Lambda+\alpha_jI)^{-2}$. The second estimate in Lemma~\ref{lem:sinc-quadrature}, followed by the Gaussian KL argument used for Theorem~\ref{thm:exact-upper-short}, proves Theorem~\ref{thm:exact-upper-independent}. See Appendix~\ref{app:proof-exact-upper-independent}.
\end{proof}

\subsection{Coordinatewise quantized smoothed-score queries}
\label{subsec:coordinate-quantized-upper}

We next give a finite-bit version of the rational sampler.  Instead of using the transformed response
$X_j=\tau_jZ+\tau_j^2s_{\tau_j}(Z)$ exactly, we quantize the weighted contribution
$W_j=c_jX_j$ coordinatewise.  Since a quantized output is discrete, we add a final isotropic Gaussian
dither.

\paragraph{Coordinatewise scalar quantizer.}
For a clipping radius $R_{\mathrm{clip}}>0$ and bit depth $B\in\mathbb N$, let $K_B=2^B$,
$\Delta_B=2R_{\mathrm{clip}}/(K_B-1)$, and
$\mathcal G_B=\{-R_{\mathrm{clip}}+\ell\Delta_B:\ell=0,\dots,K_B-1\}$.  Define
$\operatorname{clip}(t)=\max\{-R_{\mathrm{clip}},\min\{t,R_{\mathrm{clip}}\}\}$ and let
$\mathsf q_{B,R_{\mathrm{clip}}}(t)$ be a nearest point in $\mathcal G_B$ to
$\operatorname{clip}(t)$, with ties broken arbitrarily.  For $w\in\mathbb R^d$, set
\begin{align}
\mathsf Q_{B,R_{\mathrm{clip}}}(w)
=(\mathsf q_{B,R_{\mathrm{clip}}}(w_1),\dots,\mathsf q_{B,R_{\mathrm{clip}}}(w_d)).
\end{align}
Given $\delta_{\rm TV}\in(0,1)$ and $\kappa\ge 1$,
we compute the parameters as indicated in 
Algorithm~\ref{alg:coordinate-quantized-rational-sampler}.
Our main result is as follows.

\begin{algorithm}[t]
\caption{Coordinatewise-quantized rational sampler with isotropic dithering}
\label{alg:coordinate-quantized-rational-sampler}

Input: $\delta_{\rm TV}\in(0,1)$, $\kappa\ge 1$.
\begin{enumerate}
\item 
Set
$\eta
    =
\frac{\delta_{\mathrm{TV}}}{12\sqrt d}$,
and compute $h$, $M$, $N$, $\mathcal{J}$, $\alpha_j$, $c_j$ as in Section~\ref{subsec:sinc-grid}.
Set $\tau_j:=\alpha_j^{-1}$,
$q=M+N+1$,
and $\sigma^2=\frac{\delta_{\rm TV}}{12\kappa\sqrt{d}}$.
Define
$$
    R_{\mathrm{clip}}
    =
    \frac{h}{\pi}
    \sqrt{2\log\left(\frac{6dq}{\delta_{\rm TV}}\right)}.
$$
Choose $B$ to be the smallest integer satisfying
$$
    2^B-1
    \geq
    \frac{q\sqrt d\,R_{\mathrm{clip}}}{\sigma\delta_{\rm TV}}.
$$
    \item Draw $Z\sim\mathcal N(0,I_d)$.
    \item For each $j\in\mathcal J$, the \emph{provider} evaluates $s_{\tau_j}(Z)$. The \emph{encoder} forms $X_j=\tau_jZ+\tau_j^2s_{\tau_j}(Z)$ and $W_j=c_jX_j$.
    \item The \emph{encoder} computes $\widehat W_j=\mathsf Q_{B,R_{\mathrm{clip}}}(W_j)$ and sends its $dB$-bit grid indices. The \emph{sampler} reconstructs only $\widehat W_j$.
    \item Draw $G\sim\mathcal N(0,I_d)$ independently and output
    $$
        Y=\sum_{j\in\mathcal J}\widehat W_j+\sigma G.
    $$
\end{enumerate}
\end{algorithm}

\begin{theorem}[Coordinatewise quantization upper bound]
\label{thm:coordinate-quantized-upper}
For a centered Gaussian target satisfying \eqref{e_target}, run Algorithm~\ref{alg:coordinate-quantized-rational-sampler}. Then
$$
    d_{\mathrm{TV}}\left(\mathcal L(Y),\pi\right)\leq \delta_{\rm TV}.
$$
Moreover, the total number of communicated bits $Q=dBq$ satisfies
$$
    Q
    \leq
    C d
    \left(
        \log(e\kappa)+\log(e\sqrt d/\delta_{\rm TV})
    \right)
    \log(e\sqrt d/\delta_{\rm TV})
    \left(
        \log(e\kappa)+\log(ed/\delta_{\rm TV})
    \right),
$$
where $C>0$ is a universal constant.  In particular, for fixed $d$ and fixed accuracy
$\delta_{\rm TV}$, one has $Q=O(d\log^2\kappa)$.

A known mean is handled by translating the queries and output, with the same communication bound. For $\pi=\mathcal N(\mu,\Sigma)$ with the same precision bounds and a known $\|\mu\|_2\leq M_\mu$, a modified sampler attains the same TV guarantee with communication bounded by the displayed right-hand side plus
$C d\log(1+M_\mu\sqrt\kappa/\delta_{\rm TV})$; see Appendix~\ref{app:uncentered-finite-bit}.
\end{theorem}

\begin{proof}[Proof sketch]
Relative covariance control bounds rational-approximation and dither bias. Gaussian tails bound clipping, and conditioning on the independent dither turns quantization shifts into TV error. Summing these errors and counting $dB$ bits per query proves the centered case of Theorem~\ref{thm:coordinate-quantized-upper}; see Appendix~\ref{app:proof-coordinatewise}. The uncentered construction is given in Appendix~\ref{app:uncentered-finite-bit}.
\end{proof}
Theorems~\ref{thm:exact-upper-short} and~\ref{thm:exact-upper-independent} quantify query cost $q$, while Theorem~\ref{thm:coordinate-quantized-upper} quantifies communication cost $Q$. Appendix~\ref{app:implementation-error} gives a joint limit in clipping radius, bit depth, and dither variance that recovers Algorithm~\ref{alg:exact-rational-sampler}, together with a conditional additive score-error bound.

\section{A general method for  lower bounds via channel synthesis}
\label{sec:lb-channel-synthesis}

We introduce a general approach for lower bounding the query complexity using tools developed in information theory for the channel synthesis problem (also known as channel simulation or the reverse-Shannon theorem)
\citep{bennett2002entanglement,cuff2013distributed}.
Let
\(\mathcal C=\{p_x:x\in\mathcal X\}\) be a class of target distributions.  Identifying
\(p_x\) with the output law \(P_{Y|X=x}\), the class \(\mathcal C\) becomes a channel
\(P_{Y|X}\).  A sampler that receives only \(Q\) bits about \(x\) and outputs
\(\widehat Y\) with law close to \(P_{Y|X=x}\) is therefore a \(Q\)-bit simulator for
this channel.  

\paragraph{Channel synthesis game.}
Alice knows the input \(x\in\mathcal X\), while Bob must output a sample with law close
to \(P_{Y|X=x}\).  Alice and Bob share common randomness \(Z\), independent of \(x\);
one may take \(Z\) to be an infinite sequence of independent fair bits, from which all
auxiliary randomness can be generated.  Alice sends a transcript
\(T=T(x,Z)\in\mathcal T\), where \(|\mathcal T|\leq 2^Q\), and Bob outputs
\(\widehat Y=\widehat Y(T,Z)\).  We do not need separate private randomness in this
unlimited-common-randomness formulation: any such randomness can be generated from
unused coordinates of \(Z\), and revealing it to both terminals can only make simulation
easier.

\begin{figure}[t]
\centering
\begin{tikzpicture}[
    node distance=1.15cm and 1.45cm,
    box/.style={
        draw,
        rounded corners,
        align=center,
        minimum height=0.85cm,
        minimum width=1.65cm
    },
    smallbox/.style={
        draw,
        rounded corners,
        align=center,
        minimum height=0.75cm,
        minimum width=1.45cm
    },
    arr/.style={-{Latex[length=2mm]}, thick}
]

\node[box] (xtrue) {input\\$x$};
\node[box, right=of xtrue] (chan) {true channel\\$P_{Y|X}$};
\node[box, right=of chan] (ytrue) {output\\$Y$};

\draw[arr] (xtrue) -- (chan);
\draw[arr] (chan) -- (ytrue);

\node[above=0.12cm of chan] {\small Original channel};

\node[box, below=1.65cm of xtrue] (alice) {Alice\\knows $x$};
\node[smallbox, right=of alice] (msg) {transcript\\$T=T(x,Z)\in\mathcal T$\\$|\mathcal T|\le 2^Q$};
\node[box, right=of msg] (bob) {Bob\\outputs\\$\widehat Y(T,Z)$};

\node[smallbox, below=0.65cm of msg] (cr) {common randomness\\$Z\perp x$};

\draw[arr] (alice) -- (msg);
\draw[arr] (msg) -- (bob);
\draw[arr, dashed] (cr) -- (alice);
\draw[arr, dashed] (cr) -- (bob);

\node[above=0.12cm of msg] {\small Channel synthesis};

\node[align=center, below=0.25cm of bob] (tv) {
    \small \(d_{\mathrm{TV}}(P_{\widehat Y|X=x},P_{Y|X=x})\le \delta_{\mathrm{TV}}\)\\
    \small for every \(x\)
};

\end{tikzpicture}
\caption{Channel-synthesis view of finite-bit sampling.  The only information
about the input \(x\) reaching Bob is the transcript \(T=T(x,Z)\in\mathcal T\).}
\label{fig:channel-synthesis}
\end{figure}

Channel simulation is also studied as relative entropy coding \citep{li2018strong,li2024channel,flamich2023greedy}. In asymptotic memoryless settings with unlimited shared randomness, reverse-Shannon results relate the communication rate to capacity. Soft covering and rejection sampling give achievability constructions \citep{cuff2013distributed,liu2016e_,liu2018rejection}; one-shot fixed-length simulation and the sample complexity of approximate simulation require finer quantities \citep{block2023sample,flamich2024notes}. These achievability constructions assume an encoder with global knowledge of the channel input. The samplers in Theorems~\ref{thm:exact-upper-short}--\ref{thm:coordinate-quantized-upper} acquire information about the target through smoothed-score queries. Theorem~\ref{thm:channel-synthesis-converse} follows by composing a simulator with a channel decoder, following the coding--simulation duality used by \citet{bennett2014quantum}. The one-shot covariance code in Theorem~\ref{thm:capacity-random-low-rank} provides the Gaussian-specific ingredient.

\begin{theorem}[Channel-synthesis converse for TV sampling]
\label{thm:channel-synthesis-converse}
Let \(P_{Y|X}\) be a channel.  Suppose there is a simulator with common randomness
\(Z\), transcript \(T=T(X,Z)\in\mathcal T\) with \(|\mathcal T|\leq 2^Q\), and output
\(\widehat Y=\widehat Y(T,Z)\), such that for every input \(x\),
$$
    d_{\mathrm{TV}}
    \left(
        P_{\widehat Y|X=x},
        P_{Y|X=x}
    \right)
    \leq
    \delta_{\mathrm{TV}} .
$$
Let \(\epsilon(k')\) be an achievable average error probability for transmitting a
\(k'\)-bit message over the original channel \(P_{Y|X}\).  Then, for every \(k'\) with
\(\delta_{\mathrm{TV}}+\epsilon(k')<1\),
$$
    Q
    \geq
    k'
    -
    \log_2
    \frac{1}{1-\delta_{\mathrm{TV}}-\epsilon(k')}.
$$
\end{theorem}

\begin{proof}[Proof sketch]
Composing the simulator with a $k'$-bit channel decoder gives success at least $1-\delta_{\mathrm{TV}}-\epsilon(k')$. Conditional on shared randomness, $2^Q$ transcripts can decode at most $2^Q$ of the $2^{k'}$ messages correctly. Thus success is at most $2^{Q-k'}$. See Appendix~\ref{app:proof-channel-synthesis}.
\end{proof}

Theorem~\ref{thm:channel-synthesis-converse} immediately gives the following fixed-error form:
\begin{corollary}[Fixed-error version]
\label{cor:channel-synthesis-fixed-error}
Fix \(\delta_{\mathrm{TV}}\in[0,1)\), and set
\(\rho=(1-\delta_{\mathrm{TV}})/2\).  If the original channel admits a \(k'\)-bit code
with error at most \(\rho\), then every \(\delta_{\mathrm{TV}}\)-accurate simulator
satisfies
$$
    Q
    \geq
    k'
    -
    \log_2\frac{2}{1-\delta_{\mathrm{TV}}}.
$$
\end{corollary}

The duality of channel coding and channel synthesis in Theorem~\ref{thm:channel-synthesis-converse} and Corollary~\ref{cor:channel-synthesis-fixed-error}
provides a method of establishing lower bounds for channel synthesis using achievability of channel coding.
This turns an impossibility bound proof to an achievability construction.

\paragraph{Advantage of the duality approach over Fano.}
Fano-based lower bounds reduce sampling to estimating a parameter, such as the mode of the distribution \cite{lu2023upper,chewi2022query}.  
Channel synthesis
instead treats the class \(\{p_x:x\in\mathcal X\}\) directly as a channel and asks how
many bits are needed to simulate one draw from \(P_{Y|X=x}\).

The advantage is most transparent in the large-TV-error regime
\[
    \delta_{\mathrm{TV}}=1-\xi,
    \qquad
    \xi=\exp(-a d)
\]
or, equivalently up to constants in \(a\), \(\xi=2^{-ad}\).  In this regime, a
Fano-style reduction becomes essentially trivial.  Indeed, if the original channel has
a \(2^{k'}\)-message code with error at most \(\xi/2\), then a
\(\delta_{\mathrm{TV}}\)-accurate simulator induces a decoder from the simulator
transcript with success probability at least \(\xi/2\).  Fano then gives only
\(Q\geq(\xi/2)k'-h_2(\xi/2)\), losing a multiplicative factor of order
\(\xi=\exp(-ad)\).

By contrast, Corollary~\ref{cor:channel-synthesis-fixed-error} gives
\[
    Q
    \geq
    k'-\log_2(2/\xi).
\]
Thus the loss is additive rather than multiplicative.  When \(\xi=2^{-ad}\), this
additive loss is only \(ad+1\).  Theorem~\ref{thm:fixed-kappa-large-d} exploits this
gap: for fixed \(\kappa>1\) and sufficiently small \(a>0\), the channel-synthesis
converse still yields a nontrivial \(Q=\Omega_\kappa(d)\) lower bound, whereas the
corresponding Fano lower bound is exponentially suppressed.  
This idea was previously exploited by \cite{bennett2014quantum} to prove strong converses for channel coding, but here we use it to establish strong converses for channel simulation.

\paragraph{Connection to query lower bounds.}
In a finite-bit oracle model, fix the sampler's randomness and pad the response transcript to its worst-case budget $Q$. Query points, encoders, bit allocations, and stopping decisions are reconstructible from previous responses, so there are at most $2^Q$ transcripts. Theorem~\ref{thm:channel-synthesis-converse} therefore applies to adaptive queries; an expected-length bound alone would not justify this counting step.  An
ideal real-valued query is not finite-bit, but any finite-precision implementation
quantizes or noises the response.  For instance, if each coordinate of a \(d\)-dimensional
smoothed-score response is represented to precision \(\varepsilon_{\mathrm{qnt}}\)
over a bounded dynamic range, then one query reveals at most
\(O(d\log(1/\varepsilon_{\mathrm{qnt}}))\) bits.

\section{Capacity and query lower bound for the Gaussian model}
\label{sec:capacity-calculation}

In this section, 
we specialize Corollary~\ref{cor:channel-synthesis-fixed-error} to the case of the Gaussian query model, 
and derive the lower bound in Theorem~\ref{thm:gaussian-bit-lower-fixed-delta}, which is within a logarithmic factor in $\kappa$ of the upper bound in Theorem~\ref{thm:coordinate-quantized-upper} for fixed dimension and accuracy.
Let
\[
    \mathcal C_d(\kappa)
    =
    \left\{
        \mathcal N(0,\Sigma):
        \frac1\kappa I\preceq \Sigma\preceq I
    \right\}.
\]
Thus the covariance eigenvalues lie in \([1/\kappa,1]\).  We view
\(\mathcal C_d(\kappa)\) as a channel whose input is a choice of covariance matrix and
whose output is one Gaussian sample.

Let \(\epsilon_{\mathrm G}(k')\) denote the minimum achievable one-shot average decoding
error for transmitting a \(k'\)-bit message over this Gaussian covariance channel.  In
other words, a code chooses covariance matrices
\(\Sigma_1,\dots,\Sigma_{2^{k'}}\) with
\(\mathcal N(0,\Sigma_m)\in\mathcal C_d(\kappa)\), observes
\(Y\sim\mathcal N(0,\Sigma_M)\) for a uniform message \(M\), and decodes \(M\) from
\(Y\).

Let \(C_{2,d}(\kappa)\) denote the one-sample capacity in bits:
\[
    C_{2,d}(\kappa)
    =
    \sup_\Pi I(\Theta;Y).
\]
Here \(\Theta\) is a random covariance index, \(\Pi\) is an arbitrary prior distribution
over admissible covariance matrices, and
\[
    Y\mid\Theta\sim\mathcal N(0,\Sigma_\Theta),
    \qquad
    \mathcal N(0,\Sigma_\Theta)\in\mathcal C_d(\kappa).
\]
Equivalently, \(\Sigma_\Theta\) is a random covariance matrix drawn from the prior
\(\Pi\), and \(C_{2,d}(\kappa)\) is the supremum of the mutual information between this
random covariance choice and one sample drawn from the corresponding Gaussian.

Shannon's channel coding theorem states that the channel capacity \(C_{2,d}(\kappa)\) characterizes the asymptotic communication rate through a growing number of i.i.d.\ copies of the channel.
While the capacity provides good intuitions about the query complexity,
the more relevant quantity for our purpose is the one-shot channel coding error function $\epsilon_G(\cdot)$.
Theorem~\ref{thm:capacity-random-low-rank} upper bounds $\epsilon_G(\cdot)$ by a random coding argument and gives the resulting capacity lower bound.

\begin{theorem}[One-shot coding bound for covariance-bounded Gaussians]
\label{thm:capacity-random-low-rank}
For every fixed \(\rho\in(0,1)\), there are constants \(c_\rho>0\) and
\(\kappa_\rho<\infty\), independent of \(d\), such that for all \(d\ge2\) and all
\(\kappa\ge\kappa_\rho\),
$$
    \epsilon_{\mathrm G}(k')
    \leq
    \rho
    \qquad
    \text{for every }
    k'\leq c_\rho d\log_2\kappa .
$$
Consequently, there exist universal constants \(c>0\) and \(\kappa_0<\infty\) such that
for all \(d\ge2\) and all \(\kappa\ge\kappa_0\),
$$
    C_{2,d}(\kappa)
    \geq
    c\,d\log_2\kappa .
$$
\end{theorem}

\begin{proof}[Proof sketch]
Encode messages using covariance matrices with randomly oriented high-variance subspaces. A sample concentrates near its transmitted subspace; a spherical-tube bound controls confusion with a competing subspace. A union bound gives the stated one-shot code size. The capacity consequence follows from Fano's inequality. See Appendix~\ref{app:proof-random-low-rank}.
\end{proof}

Combining Theorem~\ref{thm:capacity-random-low-rank} with Corollary~\ref{cor:channel-synthesis-fixed-error} gives Theorem~\ref{thm:gaussian-bit-lower-fixed-delta}:

\begin{theorem}[Bit lower bound for condition-number-bounded Gaussians]
\label{thm:gaussian-bit-lower-fixed-delta}
Fix \(\delta_{\mathrm{TV}}\in[0,1)\). Suppose a finite-bit sampler, with arbitrary
common randomness, uses a transcript of at most \(Q\) bits and outputs \(\widehat Y\)
such that
\[
    \sup_{\pi\in\mathcal C_d(\kappa)}
    d_{\mathrm{TV}}\left(\mathcal L(\widehat Y),\pi\right)
    \leq
    \delta_{\mathrm{TV}}.
\]
Then there are constants \(c_{\delta_{\mathrm{TV}}}>0\) and
\(\kappa_{\delta_{\mathrm{TV}}}<\infty\), independent of \(d\), such that for all
\(d\ge2\) and all \(\kappa\ge\kappa_{\delta_{\mathrm{TV}}}\),
\[
    Q
    \geq
    c_{\delta_{\mathrm{TV}}}\,d\log_2\kappa .
\]
\end{theorem}

\begin{proof}
Set $\rho=(1-\delta_{\mathrm{TV}})/2$. Theorem~\ref{thm:capacity-random-low-rank} supplies a Gaussian covariance code with $k'=\lfloor c_\rho d\log_2\kappa\rfloor$ bits and error at most $\rho$ for $d\ge2$, $\kappa\ge\kappa_\rho$. A uniformly accurate sampler simulates this channel, so Corollary~\ref{cor:channel-synthesis-fixed-error} gives
\[
 Q\ge k'-\log_2\frac{2}{1-\delta_{\mathrm{TV}}}
 \ge c_\rho d\log_2\kappa-1-\log_2\frac{2}{1-\delta_{\mathrm{TV}}}.
\]
Increasing $\kappa_{\delta_{\mathrm{TV}}}$ absorbs the last two terms into half the first, uniformly over $d\ge2$. Taking $c_{\delta_{\mathrm{TV}}}=c_\rho/2$ proves Theorem~\ref{thm:gaussian-bit-lower-fixed-delta}.
\end{proof}

Theorem~\ref{thm:fixed-kappa-large-d} isolates the regime in which Corollary~\ref{cor:channel-synthesis-fixed-error} gives a visibly stronger lower bound than a Fano-style reduction.  When \(1-\delta_{\mathrm{TV}}\) is exponentially small in \(d\), Fano loses a multiplicative factor of order \(1-\delta_{\mathrm{TV}}\), whereas the channel-synthesis bound loses only an additive term of order \(\log_2(1/(1-\delta_{\mathrm{TV}}))=O(d)\).  Thus, for fixed \(\kappa>1\), channel synthesis can still yield a nontrivial linear-in-\(d\) lower bound in this large-TV-error regime.
\begin{theorem}[Fixed-\(\kappa\), large-\(d\), large-TV regime]
\label{thm:fixed-kappa-large-d}
Fix \(\kappa>1\).  There are constants \(a_\kappa,c_\kappa>0\) such that if
\(1-\delta_{\mathrm{TV}}=2^{-a d}\) with \(0<a\leq a_\kappa\), then every finite-bit
sampler satisfying
$$
    \sup_{\pi\in\mathcal C_d(\kappa)}
    d_{\mathrm{TV}}\left(\mathcal L(\widehat Y),\pi\right)
    \leq
    \delta_{\mathrm{TV}}
$$
must use
$$
    Q
    \geq
    c_\kappa d
$$
bits for all sufficiently large \(d\).
\end{theorem}

\begin{proof}
See Appendix~\ref{app:proof-fixed-kappa-large-d}.
\end{proof}

\appendix
\setcounter{figure}{0}
\renewcommand{\thefigure}{A.\arabic{figure}}
\renewcommand{\theHfigure}{appendix.\arabic{figure}}

\section{Appendix roadmap}
The appendices collect technical material omitted from the main text.
Appendix~\ref{subsec:sinc-grid} states Lemma~\ref{lem:sinc-quadrature}, which supplies the approximation used in the algorithms.
Appendices~\ref{app:numerical-sinc-validation} and~\ref{app:proof-sinc-quadrature} give the numerical validation and proof of Lemma~\ref{lem:sinc-quadrature}, respectively.
Appendix~\ref{app:proofs-upper} proves Theorems~\ref{thm:exact-upper-short}--\ref{thm:coordinate-quantized-upper} and Proposition~\ref{prop:noncentered-reduction}.
Appendix~\ref{app:proofs-channel} proves Theorem~\ref{thm:channel-synthesis-converse}.
Appendix~\ref{app:proofs-capacity} contains the proofs of Lemma~\ref{lem:spherical-tube-subspace} and Theorems~\ref{thm:capacity-random-low-rank} and~\ref{thm:fixed-kappa-large-d}.

\section{An approximation lemma}
\subsection{A sinc-quadrature grid}
\label{subsec:sinc-grid}

We state the quadrature construction as a standalone approximation result. The inputs are a scalar accuracy $0<\eta<1/2$ and an interval endpoint $\kappa\geq1$. All logarithms are natural. 
We set 
\begin{equation}
C_0:=
12/(1-e^{-1}).
    \label{eq:C0-choice}
\end{equation}
The proof of Lemma~\ref{lem:sinc-quadrature} shows that its estimates continue to hold, up to constants, when $C_0$ is a large enough constant.
Define
\begin{align}
h:=\frac{\pi^2}{\log(C_0/\eta)},\qquad
    M:=\left\lceil \frac{\log(C_0/\eta)}{h}\right\rceil,\qquad
    N:=\left\lceil \frac{\frac12\log\kappa+\log(C_0/\eta)}{h}\right\rceil .
    \label{e_hmn}
\end{align}

Let 
\begin{align}
\mathcal J=
\mathcal J_{\eta,\kappa}:=\{-M,-M+1,\dots,N\}.
\label{e_j}
\end{align}

For $j\in\mathcal J_{\eta,\kappa}$, set
$$
    \alpha_j:=e^{2jh},\qquad
    c_j:=\frac{2h}{\pi}e^{jh}.
$$
Define
$$
    r_{\eta,\kappa}(x)
    :=
    \sum_{j\in\mathcal J_{\eta,\kappa}}
    \frac{c_j}{x+\alpha_j},
    \qquad
    L_h:=\frac{2h}{\pi^2}=\frac{2}{\log(C_0/\eta)}.
$$
Thus $h$ and $M$ depend only on $\eta$, while $N$ and the index set $\mathcal J_{\eta,\kappa}$ depend on both $\eta$ and $\kappa$. The pole locations $\alpha_j$ and coefficients $c_j$ are determined by $h$ and the index $j$.

The following estimates are standard consequences of sinc quadrature, equivalently of the exponentially convergent trapezoidal rule for analytic functions on the real line; see, for example, \cite{okayama2013error,stenger1993sinc,trefethen2014trapezoidal}.

\begin{lemma}[Sinc-quadrature estimates]
\label{lem:sinc-quadrature}
With the parameters above, uniformly for all $x\in[1,\kappa]$,
\[
    \left|
        \sqrt{x}\,r_{\eta,\kappa}(x)-1
    \right|
    \leq
    \eta,
\]
and
\[
    \left|
        x
        \sum_{j\in\mathcal J_{\eta,\kappa}}
        \frac{c_j^2}{(x+\alpha_j)^2}
        -
        L_h
    \right|
    \leq
    2\eta .
\]
\end{lemma}

\begin{proof}
See Appendix~\ref{app:proof-sinc-quadrature}.
\end{proof}

\subsection{Numerical validation of the sinc-quadrature estimates}
\label{app:numerical-sinc-validation}

We include a small numerical check of Lemma~\ref{lem:sinc-quadrature}. For each pair
\((\eta,\kappa)\), define the uniform errors
\[
    E_1(\eta,\kappa)
    =
    \sup_{x\in[1,\kappa]}
    \left|
        \sqrt{x}\,r_{\eta,\kappa}(x)-1
    \right|,
\]
and
\[
    E_2(\eta,\kappa)
    =
    \sup_{x\in[1,\kappa]}
    \left|
        x
        \sum_{j\in\mathcal J_{\eta,\kappa}}
        \frac{c_j^2}{(x+\alpha_j)^2}
        -
        L_h
    \right|.
\]
We estimate these suprema on 20,001 logarithmically spaced points in $[1,\kappa]$, including both endpoints (the singleton $\{1\}$ when $\kappa=1$), using $\operatorname{sech}(t)=2e^{-|t|}/(1+e^{-2|t|})$. Code, sampled values, and reproduction instructions are included in the source package.

Figure~\ref{fig:sinc-quadrature-validation} plots sampled estimates of \(E_1(\eta,\kappa)\) and
\(E_2(\eta,\kappa)\) on log-log axes for
\[
    \kappa\in\{1,100,10000\},
    \qquad
    \log_{10}\eta\in\{-5,-4.5,\dots,-1\}.
\]
The identity line \(y=\eta\) is included as a reference.  The errors scale approximately linearly in \(\eta\), which is consistent with the \(O(\eta)\) estimates in Lemma~\ref{lem:sinc-quadrature}.  The different markers correspond to different values of \(\kappa\).  The bounded separation between the curves illustrates that the quadrature construction controls the error uniformly over the interval \([1,\kappa]\); increasing \(\kappa\) mainly increases the right truncation level \(N\) in \eqref{e_hmn}.

\begin{figure}[H]
    \centering
    \IfFileExists{sinc_quadrature_validation.pdf}{%
    \includegraphics[width=0.92\linewidth]{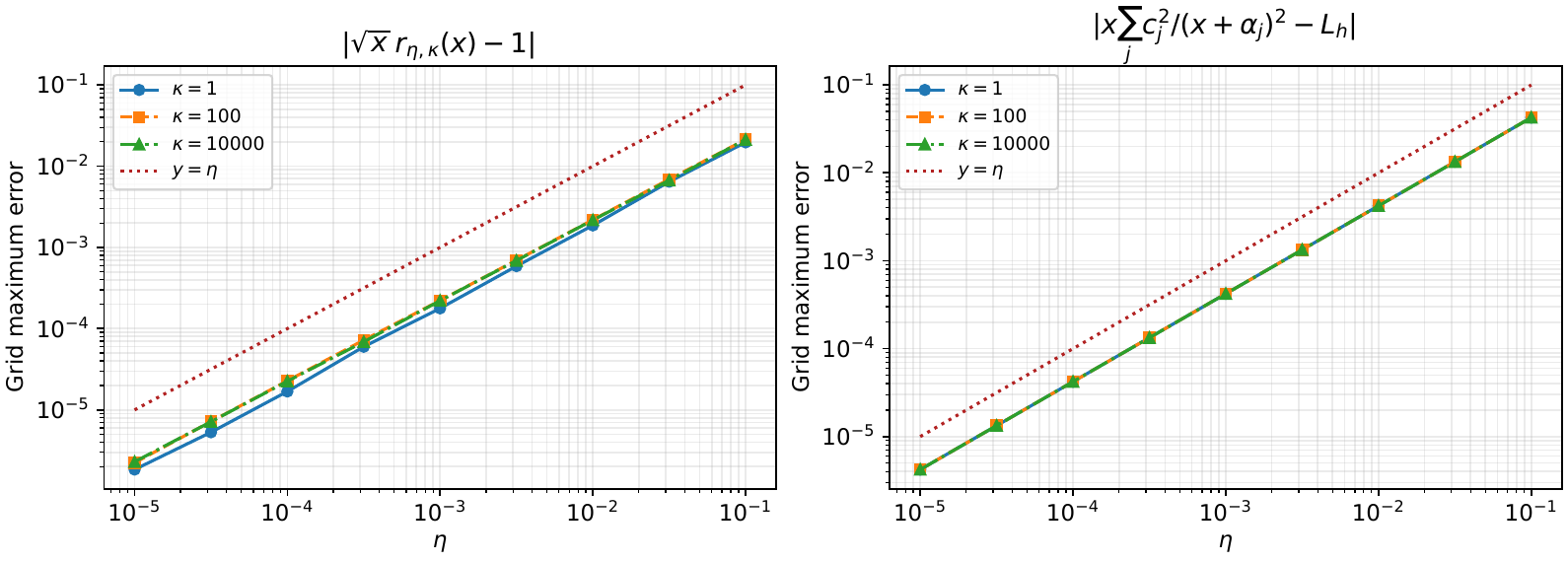}%
}{%
    \fbox{\parbox{0.88\linewidth}{\centering Missing figure file: \texttt{sinc\_quadrature\_validation.pdf}}}%
}
    \caption{Numerical validation of Lemma~\ref{lem:sinc-quadrature}: sampled estimates of $E_1(\eta,\kappa)$ (left) and $E_2(\eta,\kappa)$ (right) on a 20,001-point logarithmic grid. The dotted line is $y=\eta$.}
    \label{fig:sinc-quadrature-validation}
\end{figure}

\subsection{Proof of Lemma~\ref{lem:sinc-quadrature}}
\label{app:proof-sinc-quadrature}

Before proving Lemma~\ref{lem:sinc-quadrature}, we record the Fourier identities used below.

\begin{lemma}[Fourier transforms of \(\operatorname{sech}\) and \(\operatorname{sech}^2\)]
\label{lem:fourier-sech-identities}
Use the convention
\[
    \widehat f(\omega)
    =
    \int_{\mathbb R} f(u)e^{-i\omega u}\,du .
\]
Then
\[
    \widehat{\operatorname{sech}}(\omega)
    =
    \frac{\pi}{\cosh(\pi\omega/2)},
\]
and
\[
    \widehat{\operatorname{sech}^2}(\omega)
    =
    \frac{\pi\omega}{\sinh(\pi\omega/2)},
    \qquad
    \widehat{\operatorname{sech}^2}(0)=2 .
\]
\end{lemma}

\begin{proof}
Recall that
\[
    \cosh u=\frac{e^u+e^{-u}}{2},
    \qquad
    \operatorname{sech}u=\frac1{\cosh u}.
\]
For the first identity, use \(t=e^{2u}\). Then \(du=dt/(2t)\),
\(e^{-i\omega u}=t^{-i\omega/2}\), and
\[
    \operatorname{sech}u
    =
    \frac{2\sqrt t}{1+t}.
\]
Therefore
\[
\begin{aligned}
    \widehat{\operatorname{sech}}(\omega)
    &=
    \int_0^\infty
    t^{-i\omega/2}
    \frac{2\sqrt t}{1+t}
    \frac{dt}{2t}                                      \\
    &=
    \int_0^\infty
    \frac{t^{-1/2-i\omega/2}}{1+t}\,dt .
\end{aligned}
\]
Writing \(a=1/2-i\omega/2\), Euler's beta identity gives
\[
    \int_0^\infty \frac{t^{a-1}}{1+t}\,dt
    =
    \frac{\pi}{\sin(\pi a)}.
\]
Since
\[
    \sin\left(\frac{\pi}{2}-\frac{i\pi\omega}{2}\right)
    =
    \cosh(\pi\omega/2),
\]
we get
\[
    \widehat{\operatorname{sech}}(\omega)
    =
    \frac{\pi}{\cosh(\pi\omega/2)}.
\]

For the second identity, again use \(t=e^{2u}\). Then
\[
    \operatorname{sech}^2u
    =
    \frac{4e^{2u}}{(1+e^{2u})^2}
    =
    \frac{4t}{(1+t)^2}.
\]
Thus
\[
\begin{aligned}
    \widehat{\operatorname{sech}^2}(\omega)
    &=
    \int_0^\infty
    t^{-i\omega/2}
    \frac{4t}{(1+t)^2}
    \frac{dt}{2t}                                      \\
    &=
    2\int_0^\infty
    \frac{t^{-i\omega/2}}{(1+t)^2}\,dt .
\end{aligned}
\]
Let \(a_\omega=1-i\omega/2\). Since \(t^{-i\omega/2}=t^{a_\omega-1}\), the beta integral gives
\[
    \int_0^\infty
    \frac{t^{a_\omega-1}}{(1+t)^2}\,dt
    =
    B(a_\omega,2-a_\omega)
    =
    \Gamma(a_\omega)\Gamma(2-a_\omega).
\]
Therefore
\[
    \widehat{\operatorname{sech}^2}(\omega)
    =
    2\Gamma\left(1-\frac{i\omega}{2}\right)
     \Gamma\left(1+\frac{i\omega}{2}\right).
\]
Using
\[
    \Gamma(1+iy)\Gamma(1-iy)
    =
    \frac{\pi y}{\sinh(\pi y)},
\]
with \(y=\omega/2\), we obtain
\[
    \widehat{\operatorname{sech}^2}(\omega)
    =
    \frac{\pi\omega}{\sinh(\pi\omega/2)}.
\]
At \(\omega=0\), the right-hand side is interpreted by continuity and equals \(2\),
which also equals \(\int_{\mathbb R}\operatorname{sech}^2u\,du\).
\end{proof}

\begin{lemma}[Poisson summation identities]
\label{lem:poisson-sech-identities}
For \(h>0\) and \(a\in\mathbb R\),
\[
    \frac{h}{\pi}
    \sum_{j\in\mathbb Z}
    \operatorname{sech}(jh-a)
    =
    1+
    \sum_{m\neq0}
    \frac{e^{2\pi i m a/h}}
         {\cosh(\pi^2m/h)} .
\]
Moreover,
\[
    \frac{h^2}{\pi^2}
    \sum_{j\in\mathbb Z}
    \operatorname{sech}^2(jh-a)
    =
    \frac{2h}{\pi^2}
    +
    \sum_{m\neq0}
    e^{2\pi i m a/h}
    \frac{2m}{\sinh(\pi^2m/h)} .
\]
\end{lemma}

\begin{proof}
Let \(f(u)=\operatorname{sech}(u)\) and \(f_a(u)=f(u-a)\). Poisson summation gives
\[
    \sum_{j\in\mathbb Z} f_a(jh)
    =
    \frac1h
    \sum_{m\in\mathbb Z}
    \widehat f_a\left(\frac{2\pi m}{h}\right).
\]
Since \(f_a(u)=f(u-a)\), we have
\[
    \widehat f_a(\omega)=e^{-i\omega a}\widehat f(\omega).
\]
Using Lemma~\ref{lem:fourier-sech-identities},
\[
    \widehat f\left(\frac{2\pi m}{h}\right)
    =
    \frac{\pi}{\cosh(\pi^2m/h)}.
\]
Multiplying by \(h/\pi\) yields
\[
    \frac{h}{\pi}
    \sum_{j\in\mathbb Z}
    \operatorname{sech}(jh-a)
    =
    \sum_{m\in\mathbb Z}
    \frac{e^{-2\pi i m a/h}}
         {\cosh(\pi^2m/h)}.
\]
The \(m=0\) term is \(1\). Reindexing \(m\mapsto -m\) gives the displayed formula.

For the second identity, let \(g(u)=\operatorname{sech}^2(u)\) and \(g_a(u)=g(u-a)\).
Poisson summation gives
\[
    \frac{h^2}{\pi^2}
    \sum_{j\in\mathbb Z}g_a(jh)
    =
    \frac{h}{\pi^2}
    \sum_{m\in\mathbb Z}
    e^{-2\pi i m a/h}
    \widehat g\left(\frac{2\pi m}{h}\right).
\]
The \(m=0\) term equals
\[
    \frac{h}{\pi^2}\widehat g(0)
    =
    \frac{2h}{\pi^2}.
\]
For \(m\neq0\), Lemma~\ref{lem:fourier-sech-identities} gives
\[
\begin{aligned}
    \frac{h}{\pi^2}
    \widehat g\left(\frac{2\pi m}{h}\right)
    &=
    \frac{h}{\pi^2}
    \frac{\pi(2\pi m/h)}
         {\sinh(\pi^2m/h)}                    \\
    &=
    \frac{2m}{\sinh(\pi^2m/h)}.
\end{aligned}
\]
Reindexing \(m\mapsto -m\) gives the displayed formula.
\end{proof}

\begin{proof}[Proof of Lemma~\ref{lem:sinc-quadrature}]
For $x\in[1,\kappa]$, set
\[
    a=\frac12\log x.
\]
Then $a\in[0,\frac12\log\kappa]$. A direct calculation gives
\[
    \sqrt{x}\frac{c_j}{x+\alpha_j}
    =
    \frac{h}{\pi}\operatorname{sech}(jh-a),
\]
and
\[
    x\frac{c_j^2}{(x+\alpha_j)^2}
    =
    \frac{h^2}{\pi^2}\operatorname{sech}^2(jh-a).
\]
Therefore
\[
    \sqrt{x}\,r_{\eta,\kappa}(x)
    =
    \frac{h}{\pi}
    \sum_{j=-M}^N
    \operatorname{sech}(jh-a),
\]
and
\[
    x
    \sum_{j=-M}^N
    \frac{c_j^2}{(x+\alpha_j)^2}
    =
    \frac{h^2}{\pi^2}
    \sum_{j=-M}^N
    \operatorname{sech}^2(jh-a).
\]

We first record the corresponding infinite-grid identities. Let
\[
    f(u)=\operatorname{sech}(u),
    \qquad
    f_a(u)=f(u-a).
\]
Poisson summation gives
\[
    \sum_{j\in\mathbb Z} f_a(jh)
    =
    \frac1h
    \sum_{m\in\mathbb Z}
    \widehat f_a\left(\frac{2\pi m}{h}\right).
\]
Since \(f_a(u)=f(u-a)\), we have
\[
    \widehat f_a(\omega)=e^{-i\omega a}\widehat f(\omega).
\]
Using Lemma~\ref{lem:fourier-sech-identities},
\[
    \widehat f\left(\frac{2\pi m}{h}\right)
    =
    \frac{\pi}{\cosh(\pi^2m/h)}.
\]
Multiplying by \(h/\pi\) yields
\[
    \frac{h}{\pi}
    \sum_{j\in\mathbb Z}
    \operatorname{sech}(jh-a)
    =
    \sum_{m\in\mathbb Z}
    \frac{e^{-2\pi i m a/h}}
         {\cosh(\pi^2m/h)}.
\]
The \(m=0\) term is \(1\). Reindexing \(m\mapsto -m\) gives
\[
    \frac{h}{\pi}
    \sum_{j\in\mathbb Z}
    \operatorname{sech}(jh-a)
    =
    1+
    2\sum_{m=1}^{\infty}
    \frac{\cos(2\pi m a/h)}
         {\cosh(\pi^2m/h)}.
\]
Hence, using \(\cosh u\geq e^u/2\),
\[
\begin{aligned}
    \left|
    \frac{h}{\pi}
    \sum_{j\in\mathbb Z}
    \operatorname{sech}(jh-a)-1
    \right|
    &\leq
    2\sum_{m=1}^{\infty}
    \frac{1}{\cosh(\pi^2m/h)}
    \\
    &\leq
    4\sum_{m=1}^{\infty}
    e^{-\pi^2m/h}
    \\
    &=
    \frac{4e^{-\pi^2/h}}{1-e^{-\pi^2/h}}.
\end{aligned}
\]
Let
\[
    \rho:=e^{-\pi^2/h}.
\]
Since \(\rho\leq \eta/C_0\leq 1/C_0\), we have
\[
    \frac{4\rho}{1-\rho}
    \leq
    \frac{4\eta/C_0}{1-1/C_0}
    =
    \frac{4\eta}{C_0-1}.
\]
Because
\[
    C_0=\frac{12}{1-e^{-1}}>17,
\]
we have \(4/(C_0-1)\leq 1/4\). Therefore
\[
    \left|
    \frac{h}{\pi}
    \sum_{j\in\mathbb Z}
    \operatorname{sech}(jh-a)-1
    \right|
    \leq
    \frac{\eta}{4}.
\]

For the second identity, let
\[
    g(u)=\operatorname{sech}^2(u),
    \qquad
    g_a(u)=g(u-a).
\]
Poisson summation gives
\[
    \frac{h^2}{\pi^2}
    \sum_{j\in\mathbb Z}g_a(jh)
    =
    \frac{h}{\pi^2}
    \sum_{m\in\mathbb Z}
    e^{-2\pi i m a/h}
    \widehat g\left(\frac{2\pi m}{h}\right).
\]
The \(m=0\) term equals
\[
    \frac{h}{\pi^2}\widehat g(0)
    =
    \frac{2h}{\pi^2}.
\]
We denote
\[
    L_h:=\frac{2h}{\pi^2}=\frac{2}{\log(C_0/\eta)}.
\]
For \(m\neq0\), Lemma~\ref{lem:fourier-sech-identities} gives
\[
\begin{aligned}
    \frac{h}{\pi^2}
    \widehat g\left(\frac{2\pi m}{h}\right)
    &=
    \frac{h}{\pi^2}
    \frac{\pi(2\pi m/h)}
         {\sinh(\pi^2m/h)}
    \\
    &=
    \frac{2m}{\sinh(\pi^2m/h)}.
\end{aligned}
\]
Reindexing \(m\mapsto -m\) gives
\[
    \frac{h^2}{\pi^2}
    \sum_{j\in\mathbb Z}
    \operatorname{sech}^2(jh-a)
    =
    L_h
    +
    4\sum_{m=1}^{\infty}
    \frac{m\cos(2\pi m a/h)}
         {\sinh(\pi^2m/h)}.
\]
Since $\rho=e^{-\pi^2/h}$, for $m\geq1$ we have
\[
 \frac{1}{\sinh(\pi^2m/h)}
 =\frac{2\rho^m}{1-\rho^{2m}}
 \leq\frac{2\rho^m}{1-\rho^2}.
\]
Consequently,
\[
\begin{aligned}
    \left|
    \frac{h^2}{\pi^2}
    \sum_{j\in\mathbb Z}
    \operatorname{sech}^2(jh-a)
    -
    L_h
    \right|
    &\leq
    4\sum_{m=1}^{\infty}
    \frac{m}{\sinh(\pi^2m/h)}
    \\
    &\leq
    \frac{8}{1-\rho^2}\sum_{m=1}^{\infty}
    m e^{-\pi^2m/h}
    \\
    &=
    \frac{8e^{-\pi^2/h}}{(1-e^{-\pi^2/h})^2(1-e^{-2\pi^2/h})}
    \\
    &=
    \frac{8\rho}{(1-\rho)^2(1-\rho^2)}.
\end{aligned}
\]
Since \(\rho\leq \eta/C_0\leq 1/C_0\),
\[
    \frac{8\rho}{(1-\rho)^2(1-\rho^2)}
    \leq
    \frac{8\eta/C_0}{(1-1/C_0)^2(1-C_0^{-2})}
    =
    \frac{8C_0}{(C_0-1)^2(1-C_0^{-2})}\eta.
\]
For
\[
    C_0=\frac{12}{1-e^{-1}},
\]
one has
\[
    \frac{8C_0}{(C_0-1)^2(1-C_0^{-2})}
    \leq
    \frac12.
\]
Hence
\[
    \left|
    \frac{h^2}{\pi^2}
    \sum_{j\in\mathbb Z}
    \operatorname{sech}^2(jh-a)
    -
    L_h
    \right|
    \leq
    \frac{\eta}{2}.
\]

It remains to control the truncation error from replacing the infinite sums by
\(\sum_{j=-M}^N\). We use
\[
    \operatorname{sech}(u)\leq 2e^{-|u|},
    \qquad
    \operatorname{sech}^2(u)\leq \operatorname{sech}(u).
\]
Since \(a\in[0,\frac12\log\kappa]\), the right tail satisfies
\[
\begin{aligned}
    \frac{h}{\pi}
    \sum_{j=N+1}^{\infty}
    \operatorname{sech}(jh-a)
    &\leq
    \frac{2h}{\pi}
    \sum_{j=N+1}^{\infty}
    e^{-(jh-a)}
    \\
    &=
    \frac{2h}{\pi}
    \frac{e^a e^{-(N+1)h}}{1-e^{-h}}
    \\
    &\leq
    \frac{2h}{\pi}
    \frac{\sqrt{\kappa}\,e^{-(N+1)h}}{1-e^{-h}}.
\end{aligned}
\]
Similarly, the left tail satisfies
\[
\begin{aligned}
    \frac{h}{\pi}
    \sum_{j=-\infty}^{-M-1}
    \operatorname{sech}(jh-a)
    &\leq
    \frac{2h}{\pi}
    \sum_{j=M+1}^{\infty}
    e^{-(jh+a)}
    \\
    &\leq
    \frac{2h}{\pi}
    \frac{e^{-(M+1)h}}{1-e^{-h}}.
\end{aligned}
\]
Combining the two estimates and using the truncation choices gives
\[
\begin{aligned}
    \frac{h}{\pi}
    \sum_{j\notin[-M,N]}
    \operatorname{sech}(jh-a)
    &\leq
    \frac{2h}{\pi(1-e^{-h})}
    \left(
        \sqrt{\kappa}\,e^{-(N+1)h}
        +
        e^{-(M+1)h}
    \right)
    \\
    &\leq
    \frac{4h}{\pi(1-e^{-h})}\frac{\eta}{C_0}.
\end{aligned}
\]
We next bound the prefactor uniformly. Since \(\eta\in(0,1)\) and
\[
    h=\frac{\pi^2}{\log(C_0/\eta)}
    \leq
    \frac{\pi^2}{\log C_0}
    <4,
\]
we have
\[
    \frac{h}{1-e^{-h}}
    \leq
    \frac{4}{1-e^{-1}},
\]
where the inequality is immediate for \(h\geq1\), while for \(0<h\leq1\) it follows from
\(h/(1-e^{-h})\leq 1/(1-e^{-1})\leq 4/(1-e^{-1})\). Therefore
\[
    \frac{4h}{\pi(1-e^{-h})}\frac{\eta}{C_0}
    \leq
    \frac{16}{\pi(1-e^{-1})}\frac{\eta}{C_0}
    =
    \frac{4}{3\pi}\eta
    \leq
    \frac{\eta}{2}.
\]
Thus
\[
    \frac{h}{\pi}
    \sum_{j\notin[-M,N]}
    \operatorname{sech}(jh-a)
    \leq
    \frac{\eta}{2}.
\]
Combining this truncation estimate with the infinite-grid aliasing estimate yields
\[
\begin{aligned}
    \left|
        \sqrt{x}\,r_{\eta,\kappa}(x)-1
    \right|
    &\leq
    \left|
    \frac{h}{\pi}
    \sum_{j\in\mathbb Z}
    \operatorname{sech}(jh-a)-1
    \right|
    +
    \frac{h}{\pi}
    \sum_{j\notin[-M,N]}
    \operatorname{sech}(jh-a)
    \\
    &\leq
    \frac{\eta}{4}+\frac{\eta}{2}
    \leq
    \eta.
\end{aligned}
\]

For the squared estimate, the same truncation argument gives an even smaller tail.
Indeed, since \(\operatorname{sech}^2(u)\leq \operatorname{sech}(u)\),
\[
\begin{aligned}
    \frac{h^2}{\pi^2}
    \sum_{j\notin[-M,N]}
    \operatorname{sech}^2(jh-a)
    &\leq
    \frac{h}{\pi}
    \left[
        \frac{h}{\pi}
        \sum_{j\notin[-M,N]}
        \operatorname{sech}(jh-a)
    \right].
\end{aligned}
\]
Using again \(h\leq4\), we get
\[
    \frac{h^2}{\pi^2}
    \sum_{j\notin[-M,N]}
    \operatorname{sech}^2(jh-a)
    \leq
    \frac{4}{\pi}\cdot\frac{\eta}{2}
    \leq
    \eta.
\]
Therefore
\[
\begin{aligned}
    \left|
        \frac{h^2}{\pi^2}
        \sum_{j=-M}^N
        \operatorname{sech}^2(jh-a)
        -
        L_h
    \right|
    &\leq
    \left|
        \frac{h^2}{\pi^2}
        \sum_{j\in\mathbb Z}
        \operatorname{sech}^2(jh-a)
        -
        L_h
    \right|
    \\
    &\qquad
    +
    \frac{h^2}{\pi^2}
    \sum_{j\notin[-M,N]}
    \operatorname{sech}^2(jh-a)
    \\
    &\leq
    \frac{\eta}{2}+\eta
    \leq
    2\eta.
\end{aligned}
\]
Recalling that
\[
    x
    \sum_{j=-M}^N
    \frac{c_j^2}{(x+\alpha_j)^2}
    =
    \frac{h^2}{\pi^2}
    \sum_{j=-M}^N
    \operatorname{sech}^2(jh-a),
\]
we conclude that
\[
    \left|
        x
        \sum_{j\in\mathcal J_{\eta,\kappa}}
        \frac{c_j^2}{(x+\alpha_j)^2}
        -
        L_h
    \right|
    \leq
    2\eta.
\]
This completes the proof of Lemma~\ref{lem:sinc-quadrature}.
\end{proof}

\section{Proofs for the upper bounds}
\label{app:proofs-upper}

\subsection{Reduction to centered Gaussians}
Proposition~\ref{prop:noncentered-reduction} extends the exact-query bounds in Theorems~\ref{thm:exact-upper-short} and~\ref{thm:exact-upper-independent} to uncentered Gaussians by estimating the mean and using score differences. The finite-bit extension of Theorem~\ref{thm:coordinate-quantized-upper} is proved in Appendix~\ref{app:uncentered-finite-bit}.
\begin{proposition}[Reduction to the centered case]
\label{prop:noncentered-reduction}
Let $\pi=\mathcal N(\mu,\Sigma)$ with $0\prec\Sigma\preceq I$, put $\Lambda=\Sigma^{-1}$ and $\|v\|_\Lambda=(v^\top\Lambda v)^{1/2}$, and define
\[
    s_\tau(y)
    :=
    \nabla \log\bigl(\pi*\mathcal N(0,\tau I)\bigr)(y)
    =
    -(\Sigma+\tau I)^{-1}(y-\mu).
\]
For every $\delta_\mu>0$, there exists an estimator
$\widehat\mu$, computable using two exact smoothed-score queries, such
that
\[
    \|\widehat\mu-\mu\|_\Lambda
    \le
    \delta_\mu .
\]
\end{proposition}

\begin{proof}
We have
\[
    s_\tau(y)-s_\tau(0)
    =
    -(\Sigma+\tau I)^{-1}y
\]
which follows immediately from
\[
    s_\tau(y)
    =
    -(\Sigma+\tau I)^{-1}(y-\mu).
\]

For the mean estimation part, define
\[
    b:=s_1(0)=(\Sigma+I)^{-1}\mu .
\]
Since $\Sigma\preceq I$, all eigenvalues of $\Sigma+I$ lie in $[1,2]$,
hence
\[
    \|\mu\|_2
    =
    \|(\Sigma+I)b\|_2
    \le 2\|b\|_2 .
\]

If $b=0$, then $\mu=0$, and we may take $\widehat\mu=0$.
Otherwise, for a target accuracy $\delta_\mu>0$, define
\[
    \tau_\mu:=\frac{2\|b\|_2}{\delta_\mu},
    \qquad
    \widehat\mu:=\tau_\mu s_{\tau_\mu}(0).
\]
Since
\[
    s_{\tau_\mu}(0)
    =
    (\Sigma+\tau_\mu I)^{-1}\mu,
\]
we have
\[
    \widehat\mu-\mu
    =
    -\Sigma(\Sigma+\tau_\mu I)^{-1}\mu .
\]
Therefore,
\[
    \|\widehat\mu-\mu\|_\Lambda^2
    =
    \sum_{i=1}^d
    \frac{\lambda_i}{(\lambda_i+\tau_\mu)^2}\mu_i^2
    \le
    \frac{1}{\tau_\mu^2}\|\mu\|_2^2 ,
\]
where $\lambda_i$ are the eigenvalues of $\Sigma$.
Hence
\[
    \|\widehat\mu-\mu\|_\Lambda
    \le
    \frac{\|\mu\|_2}{\tau_\mu}
    \le
    \delta_\mu .
\]
The construction uses only the two exact smoothed-score queries
$s_1(0)$ and $s_{\tau_\mu}(0)$.
\end{proof}

The score-difference identity simulates each centered query with at most two original queries. If $\widehat X$ is the output of a centered $q$-query sampler, run with fresh randomness after estimating $\mu$, then $\widehat\mu+\widehat X$ uses at most $2q+2$ queries. The equal-covariance Gaussian bound and translation invariance give
\[
 d_{\mathrm{TV}}\!\left(\mathcal L(\widehat\mu+\widehat X),\mathcal N(\mu,\Sigma)\right)
 \leq \delta_{\mathrm{TV}}+
 \frac{\|\widehat\mu-\mu\|_\Lambda}{\sqrt{2\pi}}
 \leq\delta_{\mathrm{TV}}+\frac{\delta_\mu}{\sqrt{2\pi}}.
\]
Run the centered sampler at accuracy $\delta_{\rm TV}/2$ and choose $\delta_\mu=\sqrt{2\pi}\,\delta_{\rm TV}/2$. This proves the uncentered clauses of Theorems~\ref{thm:exact-upper-short} and~\ref{thm:exact-upper-independent} with the same asymptotic query bounds.

For finite-bit sampling, a known bound on $\|\mu\|_2$ yields the extension in Theorem~\ref{thm:coordinate-quantized-upper}; see Appendix~\ref{app:uncentered-finite-bit}.

\subsection{Proof of Theorem~\ref{thm:exact-upper-short}}
\label{app:proof-exact-upper}

\begin{proof}
For a centered Gaussian target, the smoothed score is
\[
    s_\tau(y)=-(\Sigma+\tau I)^{-1}y.
\]
If \(\alpha=1/\tau\), then
\[
    \tau y+\tau^2s_\tau(y)
    =
    (\Lambda+\alpha I)^{-1}y.
\]
Indeed, on an eigenvector of \(\Lambda\) with eigenvalue \(\lambda\), the scalar
multiplier on the left is
\[
    \tau-\frac{\tau^2}{\lambda^{-1}+\tau}
    =
    \frac{1}{\lambda+\tau^{-1}}.
\]

Therefore, for each \(j\in\mathcal J\),
\[
    X_j
    =
    (\Lambda+\alpha_j I)^{-1}Z.
\]
Hence the algorithm outputs
\[
    Y
    =
    \sum_{j\in\mathcal J}
    c_j(\Lambda+\alpha_j I)^{-1}Z
    =
    r(\Lambda)Z.
\]
By Lemma~\ref{lem:sinc-quadrature} with $r:=r_{\eta,\kappa}$,
\[
    \sup_{x\in[1,\kappa]}
    \left|
        \sqrt{x}\,r(x)-1
    \right|
    \leq
    \eta.
\]

Let \(\lambda_1,\dots,\lambda_d\) be the eigenvalues of \(\Lambda\), and define
\[
    \delta_i
    =
    \sqrt{\lambda_i}\,r(\lambda_i)-1.
\]
Then \(|\delta_i|\leq\eta\). Since
\[
    Y\sim\mathcal N(0,r(\Lambda)^2)
    \qquad
    \text{and}
    \qquad
    \pi=\mathcal N(0,\Lambda^{-1}),
\]
the relative covariance eigenvalues are
\[
    \lambda_i r(\lambda_i)^2
    =
    (1+\delta_i)^2.
\]
Thus
\[
    \operatorname{KL}
    \left(
        \mathcal L(Y)
        \,\middle\|\,
        \pi
    \right)
    =
    \frac12
    \sum_{i=1}^{d}
    \left[
        (1+\delta_i)^2-1-2\log(1+\delta_i)
    \right].
\]
For \(|u|\leq1/2\),
\[
    (1+u)^2-1-2\log(1+u)
    \leq
    3u^2.
\]
Since the algorithm sets
\[
    \eta
    =
    \frac{\delta_{\mathrm{TV}}}{4\sqrt d},
\]
we have \(\eta\leq1/2\), and therefore
\[
    \operatorname{KL}
    \left(
        \mathcal L(Y)
        \,\middle\|\,
        \pi
    \right)
    \leq
    \frac32 d\eta^2.
\]
By Pinsker's inequality,
\[
\begin{aligned}
    d_{\mathrm{TV}}
    \left(
        \mathcal L(Y),
        \pi
    \right)
    &\leq
    \sqrt{
        \frac12
        \operatorname{KL}
        \left(
            \mathcal L(Y)
            \,\middle\|\,
            \pi
        \right)
    }                                      \\
    &\leq
    \frac{\sqrt3}{2}\sqrt d\,\eta
    \leq
    \sqrt d\,\eta
    \leq
    \delta_{\mathrm{TV}}.
\end{aligned}
\]

Finally, since \(h,M,N\) are defined in \eqref{e_hmn},
\[
    q=M+N+1
    =
    O\left(
        \left(\log\kappa+\log(1/\eta)\right)
        \log(1/\eta)
    \right).
\]
Substituting \(\eta=\delta_{\mathrm{TV}}/(4\sqrt d)\) gives
\[
    q
    =
    O\left(
        \left(
            \log\kappa+\log(e\sqrt d/\delta_{\mathrm{TV}})
        \right)
        \log(e\sqrt d/\delta_{\mathrm{TV}})
    \right).
\]
This proves Theorem~\ref{thm:exact-upper-short}.
\end{proof}

\subsection{Proof of Theorem~\ref{thm:exact-upper-independent}}
\label{app:proof-exact-upper-independent}
\begin{proof}
For a Gaussian target, the smoothed score satisfies
\[
    s_\tau(y)=-(\Sigma+\tau I)^{-1}y.
\]
Since $\tau_j=\alpha_j^{-1}$, the query output obeys
\[
    X_j
    =
    \tau_jZ_j+\tau_j^2s_{\tau_j}(Z_j)
    =
    (\Lambda+\alpha_jI)^{-1}Z_j.
\]
Let
\[
    R_j=(\Lambda+\alpha_jI)^{-1}.
\]
Then Algorithm~\ref{alg:exact-rational-sampler-independent} outputs
\[
    Y
    =
    \frac{1}{\sqrt{L_h}}
    \sum_{j\in\mathcal J} c_j R_j Z_j,
\]
where the random vectors $(Z_j)_{j\in\mathcal J}$ are independent standard Gaussians.

It follows that $Y$ is centered Gaussian with covariance
\[
    \cov(Y)
    =
    \frac{1}{L_h}
    \sum_{j\in\mathcal J}
    c_j^2(\Lambda+\alpha_jI)^{-2}.
\]
Since this covariance matrix is a function of $\Lambda$, it is diagonal in the same eigenbasis
as the target covariance $\Lambda^{-1}$. If $\lambda\in[1,\kappa]$ is an eigenvalue of
$\Lambda$, then the variance of $Y$ in that eigendirection is
\[
    \frac{1}{L_h}
    \sum_{j\in\mathcal J}
    \frac{c_j^2}{(\lambda+\alpha_j)^2}.
\]
The corresponding covariance ratio relative to the target variance $1/\lambda$ is therefore
\[
    T(\lambda)
    =
    \frac{\lambda}{L_h}
    \sum_{j\in\mathcal J}
    \frac{c_j^2}{(\lambda+\alpha_j)^2}.
\]
By the second estimate in Lemma~\ref{lem:sinc-quadrature},
\[
    |T(\lambda)-1|
    \leq
    \frac{C\eta}{L_h}.
\]
In Lemma~\ref{lem:sinc-quadrature} we may take $C=2$. Thus
\[
    |T(\lambda)-1|
    \leq
    \rho,
    \qquad
    \rho:=\frac{2\eta}{L_h}.
\]

We now verify that the explicit choice of $\eta$ implies the desired bound on $\rho$.
Recall that
\[
    L_h=\frac{2h}{\pi^2},
    \qquad
    h=\frac{\pi^2}{\log(C_0/\eta)},
\]
and hence
\[
    L_h=\frac{2}{\log(C_0/\eta)}.
\]
Therefore
\[
    \rho
    =
    \frac{2\eta}{L_h}
    =
    \eta\log(C_0/\eta).
\]
Let
\[
    r:=\frac{\delta_{\rm TV}}{\sqrt d},
    \qquad
    L:=\log(C_0/r)
    =
    \log(C_0\sqrt d/\delta_{\rm TV}).
\]
The algorithm chooses
\[
    \eta
    =
    \frac{r}{8L}
    =
    \frac{\delta_{\rm TV}}
    {8\sqrt d\,\log(C_0\sqrt d/\delta_{\rm TV})}.
\]
Hence
\[
    \log(C_0/\eta)
    =
    \log\left(\frac{8C_0L}{r}\right)
    =
    L+\log(8L).
\]
Consequently,
\[
\begin{aligned}
    \rho
    =
    \eta\log(C_0/\eta)
    &=
    \frac{r}{8L}
    \left(L+\log(8L)\right)        \\
    &=
    \frac r8
    \left(
        1+\frac{\log(8L)}{L}
    \right).
\end{aligned}
\]
Since $\delta_{\rm TV}\in(0,1)$ and $d\geq1$, we have $r\leq1$. Moreover,
\[
    L=\log(C_0/r)\geq \log C_0.
\]
With $C_0=12/(1-e^{-1})$, $\log C_0>2$, and therefore
\[
    1+\frac{\log(8L)}{L}\leq 4 .
\]
It follows that
\[
    \rho
    \leq
    \frac r2
    =
    \frac{\delta_{\rm TV}}{2\sqrt d}.
\]
In particular, since $\delta_{\rm TV}\in(0,1)$ and $d\geq1$,
\[
    \rho\leq \frac12.
\]

Let $\lambda_1,\ldots,\lambda_d$ denote the eigenvalues of $\Lambda$, and write
$T_i:=T(\lambda_i)$. The Gaussian KL formula gives
\[
    D\left(\mathcal L(Y)\,\middle\|\,\pi\right)
    =
    \frac12
    \sum_{i=1}^d
    \left(
        T_i-1-\log T_i
    \right).
\]
Since $|T_i-1|\leq \rho\leq 1/2$, the elementary inequality
\[
    u-\log(1+u)\leq u^2,
    \qquad |u|\leq 1/2,
\]
with $u=T_i-1$, yields
\[
    D\left(\mathcal L(Y)\,\middle\|\,\pi\right)
    \leq
    \frac12 d\rho^2.
\]
By Pinsker's inequality,
\[
    d_{\mathrm{TV}}\left(\mathcal L(Y),\pi\right)
    \leq
    \sqrt{\frac12
    D\left(\mathcal L(Y)\,\middle\|\,\pi\right)}
    \leq
    \frac{\sqrt d}{2}\rho
    \leq
    \frac{\delta_{\rm TV}}{4}
    \leq
    \delta_{\rm TV}.
\]

It remains to record the query complexity. The number of queries is
\[
    q=|\mathcal J|=M+N+1.
\]
With the sinc-grid parameters defined in Section~\ref{subsec:sinc-grid}, and with
\[
    \eta
    =
    \frac{\delta_{\rm TV}}
    {8\sqrt d\,\log(C_0\sqrt d/\delta_{\rm TV})},
\]
we have
\[
    \log(1/\eta)
    =
    O\left(\log(e\sqrt d/\delta_{\rm TV})\right).
\]
Indeed, the extra factor
$\log(C_0\sqrt d/\delta_{\rm TV})$ inside $1/\eta$ only contributes an additive
$\log\log(C_0\sqrt d/\delta_{\rm TV})$, which is absorbed by
$\log(e\sqrt d/\delta_{\rm TV})$. Therefore the same estimates as in
Theorem~\ref{thm:exact-upper-short} give
\[
    q
    =
    O\left(
        \left(
            \log\kappa+\log(e\sqrt d/\delta_{\rm TV})
        \right)
        \log(e\sqrt d/\delta_{\rm TV})
    \right),
\]
where $O(\cdot)$ hides universal constants. This completes the proof of Theorem~\ref{thm:exact-upper-independent}.
\end{proof}

\subsection{Centered case of Theorem~\ref{thm:coordinate-quantized-upper}}
\label{app:proof-coordinatewise}

\begin{proof}
For the centered case, the smoothed score is
\[
    s_\tau(y)=-(\Sigma+\tau I)^{-1}y.
\]
Since \(\tau_j=\alpha_j^{-1}\), we have
\[
    X_j
    =
    \tau_jZ+\tau_j^2s_{\tau_j}(Z)
    =
    (\Lambda+\alpha_jI)^{-1}Z.
\]
Writing \(R_j=(\Lambda+\alpha_jI)^{-1}\), the weighted contribution is
\[
    W_j=c_jX_j=c_jR_jZ.
\]
The exact unquantized output is
\[
    Y_0
    =
    \sum_{j\in\mathcal J}W_j
    =
    r(\Lambda)Z,
    \qquad
    r(x)=\sum_{j\in\mathcal J}\frac{c_j}{x+\alpha_j}.
\]
The actual quantized and dithered output is
\[
    Y
    =
    Y_0+E_{\mathrm{qnt}}+\sigma G,
    \qquad
    E_{\mathrm{qnt}}
    =
    \sum_{j\in\mathcal J}(\widehat W_j-W_j),
\]
where \(G\sim\mathcal N(0,I_d)\) is independent of all other randomness.

First compare with the ideal dithered output
\[
    Y_\sigma^\star
    =
    Y_0+\sigma G.
\]
Then, with $r:=r_{\eta,\kappa}$,
\[
    Y_\sigma^\star
    \sim
    \mathcal N(0,r(\Lambda)^2+\sigma^2I).
\]
Since both \(r(\Lambda)^2+\sigma^2I\) and \(\Lambda^{-1}\) are functions of
\(\Lambda\), they are diagonal in the same eigenbasis. If
\(\lambda\in[1,\kappa]\) is an eigenvalue of \(\Lambda\), then the target variance in
the corresponding eigendirection is \(1/\lambda\), while the variance of
\(Y_\sigma^\star\) is \(r(\lambda)^2+\sigma^2\). Hence the covariance ratio in this
direction is
\[
    T_\sigma(\lambda)
    =
    \lambda(r(\lambda)^2+\sigma^2)
    =
    \left(\sqrt{\lambda}\,r(\lambda)\right)^2+\sigma^2\lambda.
\]
By Lemma~\ref{lem:sinc-quadrature},
\[
    \left|\sqrt{\lambda}\,r(\lambda)-1\right|
    \leq
    \eta.
\]
Since \(\eta<1/2\),
\[
    \left|
        \left(\sqrt{\lambda}\,r(\lambda)\right)^2-1
    \right|
    \leq
    3\eta.
\]
Also \(\lambda\leq\kappa\), so
\[
    |T_\sigma(\lambda)-1|
    \leq
    3\eta+\kappa\sigma^2.
\]
Set
\[
    \rho:=3\eta+\kappa\sigma^2 .
\]
By the parameter choices in the algorithm,
\[
    \eta=\frac{\delta_{\rm TV}}{12\sqrt d},
    \qquad
    \sigma^2=\frac{\delta_{\rm TV}}{12\kappa\sqrt d},
\]
and therefore
\[
    \rho
    =
    \frac{\delta_{\rm TV}}{4\sqrt d}
    +
    \frac{\delta_{\rm TV}}{12\sqrt d}
    =
    \frac{\delta_{\rm TV}}{3\sqrt d}
    \leq
    \frac12.
\]

Let \(\lambda_1,\ldots,\lambda_d\) be the eigenvalues of \(\Lambda\), and write
\(T_i=T_\sigma(\lambda_i)\). Then \(|T_i-1|\leq\rho\) for all \(i\). The Gaussian KL
formula gives
\[
    D\left(\mathcal L(Y_\sigma^\star)\,\middle\|\,\pi\right)
    =
    \frac12
    \sum_{i=1}^d
    \left(
        T_i-1-\log T_i
    \right).
\]
Indeed, \(T_i\) is exactly the ratio between the variance of \(Y_\sigma^\star\) and
the target variance in the \(i\)-th eigendirection. Since \(T_i=1+u_i\) with
\(|u_i|\leq\rho\leq1/2\), the elementary inequality
\[
    u-\log(1+u)\leq u^2,
    \qquad |u|\leq \frac12,
\]
implies
\[
    D\left(\mathcal L(Y_\sigma^\star)\,\middle\|\,\pi\right)
    \leq
    \frac12\sum_{i=1}^d u_i^2
    \leq
    \frac12 d\rho^2 .
\]
By Pinsker's inequality,
\[
\begin{aligned}
    d_{\mathrm{TV}}\left(\mathcal L(Y_\sigma^\star),\pi\right)
    &\leq
    \sqrt{
        \frac12
        D\left(\mathcal L(Y_\sigma^\star)\,\middle\|\,\pi\right)
    }                                      \\
    &\leq
    \frac{\sqrt d}{2}\rho                  \\
    &=
    \frac{\sqrt d}{2}
    \left(3\eta+\kappa\sigma^2\right)
    =
    \frac{\delta_{\rm TV}}{6}.
\end{aligned}
\]

It remains to control the quantization error. For \(t=e^{jh}\), we have
\(\alpha_j=t^2\) and \(c_j=(2h/\pi)t\). Hence
\[
    \|c_jR_j\|_{\mathrm{op}}
    =
    \sup_{\lambda\in[1,\kappa]}
    \frac{c_j}{\lambda+\alpha_j}
    \leq
    \frac{2h}{\pi}\frac{t}{1+t^2}
    \leq
    \frac{h}{\pi}.
\]
Thus every coordinate of every \(W_j=c_jR_jZ\) is a centered Gaussian with variance at
most \((h/\pi)^2\). Hence, for each \(j\in\mathcal J\) and each coordinate
\(\ell\in[d]\),
\[
    \mathbb P\left(
        |(W_j)_\ell|>R_{\mathrm{clip}}
    \right)
    \leq
    2\exp\left(
        -\frac{\pi^2R_{\mathrm{clip}}^2}{2h^2}
    \right).
\]
By a union bound over all \(q=|\mathcal J|\) values of \(j\) and all \(d\)
coordinates,
\[
    \mathbb P\left(
        \max_{j\in\mathcal J}
        \max_{\ell\in[d]}
        |(W_j)_\ell|>R_{\mathrm{clip}}
    \right)
    \leq
    2dq
    \exp\left(
        -\frac{\pi^2R_{\mathrm{clip}}^2}{2h^2}
    \right).
\]
The algorithm sets
\[
    p_{\rm clip}
    =
    \frac{\delta_{\rm TV}}{3},
    \qquad
    R_{\mathrm{clip}}
    =
    \frac{h}{\pi}
    \sqrt{
        2\log\left(\frac{2dq}{p_{\rm clip}}\right)
    }
    =
    \frac{h}{\pi}
    \sqrt{
        2\log\left(\frac{6dq}{\delta_{\rm TV}}\right)
    }.
\]
Therefore
\[
    \frac{\pi^2R_{\mathrm{clip}}^2}{2h^2}
    =
    \log\left(\frac{2dq}{p_{\rm clip}}\right),
\]
and hence
\[
\begin{aligned}
    2dq
    \exp\left(
        -\frac{\pi^2R_{\mathrm{clip}}^2}{2h^2}
    \right)
    &=
    2dq
    \exp\left(
        -\log\left(\frac{2dq}{p_{\rm clip}}\right)
    \right) \\
    &=
    p_{\rm clip}
    =
    \frac{\delta_{\rm TV}}{3}.
\end{aligned}
\]
Thus the clipping event fails with probability at most \(p_{\rm clip}\).

On the no-clipping event, each coordinate quantization error is at most
\[
    \frac{\Delta_B}{2}
    =
    \frac{R_{\mathrm{clip}}}{2^B-1}.
\]
Therefore
\[
    \|E_{\mathrm{qnt}}\|_2
    \leq
    \sum_{j\in\mathcal J}
    \|\widehat W_j-W_j\|_2
    \leq
    \frac{q\sqrt d\,R_{\mathrm{clip}}}{2^B-1}.
\]
Condition on \(Z\). Then \(Y\) and \(Y_\sigma^\star\) differ only by the deterministic
shift \(E_{\mathrm{qnt}}\), and both have common Gaussian dither covariance
\(\sigma^2I\). For equal-covariance Gaussians,
\[
    d_{\mathrm{TV}}
    \left(
        \mathcal N(m+e,\sigma^2I),
        \mathcal N(m,\sigma^2I)
    \right)
    \leq
    \frac{\|e\|_2}{2\sigma}.
\]
By convexity of total variation under mixtures, splitting according to the clipping
event gives
\[
    d_{\mathrm{TV}}
    \left(
        \mathcal L(Y),
        \mathcal L(Y_\sigma^\star)
    \right)
    \leq
    p_{\rm clip}
    +
    \frac{q\sqrt d\,R_{\mathrm{clip}}}{2\sigma(2^B-1)}.
\]
By the choice of \(B\),
\[
    2^B-1
    \geq
    \frac{q\sqrt d\,R_{\mathrm{clip}}}{\sigma\delta_{\rm TV}},
\]
so
\[
    \frac{q\sqrt d\,R_{\mathrm{clip}}}{2\sigma(2^B-1)}
    \leq
    \frac{\delta_{\rm TV}}{2}.
\]
Since \(p_{\rm clip}=\delta_{\rm TV}/3\), we obtain
\[
    d_{\mathrm{TV}}
    \left(
        \mathcal L(Y),
        \mathcal L(Y_\sigma^\star)
    \right)
    \leq
    \frac{\delta_{\rm TV}}{3}
    +
    \frac{\delta_{\rm TV}}{2}.
\]
Combining this with
\[
    d_{\mathrm{TV}}\left(\mathcal L(Y_\sigma^\star),\pi\right)
    \leq
    \frac{\delta_{\rm TV}}{6}
\]
gives
\[
    d_{\mathrm{TV}}\left(\mathcal L(Y),\pi\right)
    \leq
    \delta_{\rm TV}.
\]

Finally, since \(\eta=\delta_{\rm TV}/(12\sqrt d)\) and \(h,M,N\) are defined by
\eqref{e_hmn},
\[
    q=M+N+1
    =
    O\left(
        \left(
            \log(e\kappa)+\log(e\sqrt d/\delta_{\rm TV})
        \right)
        \log(e\sqrt d/\delta_{\rm TV})
    \right).
\]
The choice of \(B\) gives
\[
\begin{aligned}
    B
    &=
    O\left(
        1+\log_2\left(
            \frac{q\sqrt d\,R_{\mathrm{clip}}}{\sigma\delta_{\rm TV}}
        \right)
    \right) \\
    &=
    O\left(
        \log(e\kappa)+\log(ed/\delta_{\rm TV})
    \right),
\end{aligned}
\]
using
\[
    R_{\mathrm{clip}}
    =
    \frac{h}{\pi}
    \sqrt{2\log(6dq/\delta_{\rm TV})},
    \qquad
    \sigma
    =
    \sqrt{
        \frac{\delta_{\rm TV}}{12\kappa\sqrt d}
    }.
\]
Since each query communicates \(dB\) bits, \(Q=dBq\), and the centered communication bound in Theorem~\ref{thm:coordinate-quantized-upper} follows.
\end{proof}

\subsection{Implementation error and the exact-query limit}
\label{app:implementation-error}
The centered proof of Theorem~\ref{thm:coordinate-quantized-upper} separates clipping probability from the error on the no-clipping event. Its basic smoothing estimate is useful more generally. Let $G\sim\mathcal N(0,I_d)$ be independent of a jointly distributed pair $(U,E)$, with $\mathbb E\|E\|_2<\infty$, and let $\sigma>0$. Conditioning on $(U,E)$, using the equal-covariance Gaussian bound, and averaging gives
\[
 d_{\mathrm{TV}}\!\left(\mathcal L(U+E+\sigma G),\mathcal L(U+\sigma G)\right)
 \leq \frac{\mathbb E\|E\|_2}{2\sigma}.
\]
In particular, for the unquantized weighted construction, replace $s_{\tau_j}(Z)$ by $s_{\tau_j}(Z)+e_j$. The resulting additive output error is
\[
 E_{\rm score}=\sum_{j\in\mathcal J}c_j\tau_j^2e_j.
\]
If $\mathbb E\|e_j\|_2\leq\varepsilon_j$ and the final dither is independent of $Z$ and all score errors, then
\[
 d_{\mathrm{TV}}\!\left(\mathcal L(U+E_{\rm score}+\sigma G),\mathcal L(U+\sigma G)\right)
 \leq \frac1{2\sigma}\sum_{j\in\mathcal J}|c_j|\tau_j^2\varepsilon_j,
 \qquad U=r(\Lambda)Z.
\]
The score errors may depend on $Z$ and on one another. This estimate controls the score-error contribution under the stated assumptions. A full sampling guarantee also requires bounds on rational-approximation error, dither bias, clipping, and quantization. Relative score error additionally requires estimates of the response magnitudes at each level. Extending the analysis to non-Gaussian scores requires further stability assumptions.

There is also a precise connection between Algorithms~\ref{alg:coordinate-quantized-rational-sampler} and~\ref{alg:exact-rational-sampler}. Fix the rational grid, and let $Y_0=r(\Lambda)Z$. For the same quantization construction with variable $R_{\mathrm{clip}},B,\sigma$, the proof of Theorem~\ref{thm:coordinate-quantized-upper} gives
\[
 d_{\mathrm{TV}}\!\left(\mathcal L(Y),\mathcal L(Y_0+\sigma G)\right)
 \leq \mathbb P\!\left(\max_{j,\ell}|(W_j)_\ell|>R_{\mathrm{clip}}\right)
 +\frac{q\sqrt d\,R_{\mathrm{clip}}}{2\sigma(2^B-1)}.
\]
Consequently, if $R_{\mathrm{clip}}\to\infty$, $\sigma\to0$, and $R_{\mathrm{clip}}2^{-B}/\sigma\to0$, both terms vanish. Since $r(\Lambda)^2$ is positive definite, Gaussian KL also shows $Y_0+\sigma G\to Y_0$ in total variation. Thus this joint limit recovers the output law of Algorithm~\ref{alg:exact-rational-sampler} for the fixed grid. Increasing $B$ alone at fixed clipping radius and dither does not do so.

\subsection{Uncentered finite-bit sampling}
\label{app:uncentered-finite-bit}
\begin{proof}[Proof of the uncentered extension in Theorem~\ref{thm:coordinate-quantized-upper}]
If $\mu$ is known, query at $\mu+y$ in place of $y$ and add $\mu$ to the output. Since $s_\tau(\mu+y)=-(\Sigma+\tau I)^{-1}y$, this gives the centered construction with unchanged query and communication costs.

Suppose now that $\|\mu\|_2\leq M_\mu$ is known. The case $M_\mu=0$ is centered, so assume $M_\mu>0$. Write $\delta=\delta_{\rm TV}$ and choose the parameters of Algorithm~\ref{alg:coordinate-quantized-rational-sampler} at accuracy $\varepsilon=\delta/4$. In particular, retain its grid, $q$, $R_{\mathrm{clip}}$, $\sigma$, and $B$. Define
\[
 S=\sum_{j\in\mathcal J}c_j\tau_j,
 \qquad
 a=\max\left\{1,\frac{M_\mu S}{R_{\mathrm{clip}}},
                 \frac{2M_\mu S}{\sigma\delta}\right\}.
\]
Draw $Z\sim\mathcal N(0,I_d)$ and query at $aZ$. The encoder forms
\[
\begin{aligned}
 W_j^{(a)}
 &=c_j\left[\tau_jZ+\frac{\tau_j^2}{a}s_{\tau_j}(aZ)\right]
   =W_j+v_j,\\
 W_j&=c_j(\Lambda+\alpha_jI)^{-1}Z,
 \qquad
 v_j=\frac{c_j\tau_j^2}{a}(\Sigma+\tau_jI)^{-1}\mu.
\end{aligned}
\]
The coefficients $c_j$ are positive, and
\[
 \|v_j\|_2\leq\frac{M_\mu c_j\tau_j}{a}
 \leq R_{\mathrm{clip}},
 \qquad
 \left\|\sum_jv_j\right\|_2
 \leq\frac{M_\mu S}{a}\leq\frac{\sigma\delta}{2}.
\]
Quantize $W_j^{(a)}$ coordinatewise with clipping radius $2R_{\mathrm{clip}}$ and bit depth $B+1$, obtaining $\widehat W_j^{(a)}$. The centered no-clipping event
$\max_{j,\ell}|(W_j)_\ell|\leq R_{\mathrm{clip}}$ has failure probability at most $\varepsilon/3$, by Appendix~\ref{app:proof-coordinatewise}. On that event, all $W_j^{(a)}$ lie within the enlarged clipping range. Since
\[
 \frac{2R_{\mathrm{clip}}}{2^{B+1}-1}
 \leq\frac{R_{\mathrm{clip}}}{2^B-1},
\]
the quantization contribution to TV remains at most $\varepsilon/2$. With an independent $G\sim\mathcal N(0,I_d)$, the deterministic offset contributes at most
$\|\sum_jv_j\|_2/(2\sigma)\leq\delta/4$.
The rational-approximation and dither bias is at most $\varepsilon/6$. Thus the covariance-stage output
\[
 Y_{\rm cov}=\sum_j\widehat W_j^{(a)}+\sigma G
\]
satisfies
\[
 d_{\mathrm{TV}}\!\left(\mathcal L(Y_{\rm cov}),\mathcal N(0,\Sigma)\right)
 \leq\frac\varepsilon3+\frac\varepsilon2+\frac\varepsilon6+\frac\delta4
 =\frac\delta2.
\]

For the mean, make one further query at zero with
$\tau_0=\max\{1,4M_\mu/\delta\}$ and let $\mu_{\tau_0}=\tau_0s_{\tau_0}(0)$. Then
\[
 \|\mu_{\tau_0}\|_2\leq M_\mu,
 \qquad
 \|\mu_{\tau_0}-\mu\|_\Lambda
 \leq\frac{M_\mu}{\tau_0}\leq\frac\delta4.
\]
Fix a Euclidean $\delta/(4\sqrt\kappa)$-net of the radius-$M_\mu$ ball, shared by the encoder and sampler. A maximal separated net has cardinality at most
$(1+8M_\mu\sqrt\kappa/\delta)^d$. The encoder sends a nearest-net index for $\mu_{\tau_0}$, using at most
\[
 Q_\mu=\left\lceil d\log_2\left(1+\frac{8M_\mu\sqrt\kappa}{\delta}\right)\right\rceil
\]
bits. The decoded vector $\widehat\mu$ satisfies
$\|\widehat\mu-\mu\|_\Lambda\leq\delta/2$.
Consequently, the final output $Y=\widehat\mu+Y_{\rm cov}$ obeys
\[
 d_{\mathrm{TV}}\!\left(\mathcal L(Y),\mathcal N(\mu,\Sigma)\right)
 \leq\frac\delta2+\frac{\delta}{2\sqrt{2\pi}}
 \leq\delta.
\]
The total communication is $dq(B+1)+Q_\mu$. The bounds on $q$ and $B$ from Theorem~\ref{thm:coordinate-quantized-upper}, with $\varepsilon=\delta/4$, give the stated communication bound: the extra $dq$ bits and the ceiling in $Q_\mu$ are absorbed by the displayed centered bound, and the remaining mean cost is
$O(d\log(1+M_\mu\sqrt\kappa/\delta))$.
The construction uses at most $q+1$ queries. Its query amplitudes and descriptions are uncharged in the oracle model.
\end{proof}

A bound on an unknown mean is necessary for a uniform finite-bit budget. Even with $\Sigma=I$, the means $\mu_\ell=\ell h e_1$, $1\leq\ell\leq L$, form arbitrarily large one-shot codes: for $L=2^{k\prime}$, choose $h$ so that nearest-mean decoding has error at most $\rho=(1-\delta_{\rm TV})/2$. Theorem~\ref{thm:channel-synthesis-converse} gives $Q\geq\log_2 L-\log_2(2/(1-\delta_{\rm TV}))$, so no uniform finite-bit budget covers arbitrary unknown means.

\section{Proof of the channel-synthesis converse}
\label{app:proofs-channel}
\subsection{Proof of Theorem~\ref{thm:channel-synthesis-converse}}
\label{app:proof-channel-synthesis}
\begin{proof}
Fix a code for $P_{Y|X}$ with $L=2^{k'}$ messages and error probability $\epsilon(k')$. Thus there are an encoder $x(m)$ and a decoder $\varphi(y)$ such that, for $M\sim\operatorname{Unif}([L])$ and $Y\sim P_{Y|X=x(M)}$,
$$
    \mathbb P(\varphi(Y)\neq M)
    \leq
    \epsilon(k').
$$

Now replace the true channel output $Y$ by the simulated output $\widehat Y$. For each message $m$, the channel input is $x(m)$. The simulation guarantee holds for every input $x$, and therefore
$$
    d_{\mathrm{TV}}
    \left(
        P_{\widehat Y|M=m},
        P_{Y|M=m}
    \right)
    \leq
    \delta_{\mathrm{TV}}
    \qquad
    \text{for every }m\in[L].
$$
The probability above is over the common randomness $Z$. Hence the same decoder satisfies
$$
    \mathbb P(\varphi(\widehat Y)\neq M)
    \leq
    \epsilon(k')+\delta_{\mathrm{TV}} .
$$
Equivalently, $\widetilde M=\varphi(\widehat Y)$ estimates $M$ with success probability at least
$$
    \mathbb P(\widetilde M=M)
    \geq
    1-\epsilon(k')-\delta_{\mathrm{TV}} .
$$

After replacing $Y$ by $\widehat Y$, the true channel output is no longer available. The simulated output is a function of the transcript and common randomness, $\widehat Y=\widehat Y(T,Z)$, so $\widetilde M=\varphi(\widehat Y)$ is also a function of $(T,Z)$. Thus there is an estimator $\psi$ such that $\widetilde M=\psi(T,Z)$.

We claim that every estimator $\psi(T,Z)$ satisfies
$$
    \mathbb P(\psi(T,Z)=M)
    \leq
    \frac{2^Q}{L}
    =
    2^{Q-k'}.
$$
To see this, condition on a realization \(Z=z\). Since the common randomness \(Z\) is independent of the message \(M\), the conditional distribution of \(M\) given \(Z=z\) is still uniform on \([L]\). The transcript always takes values in a fixed alphabet \(\mathcal{T}\) with \(|\mathcal{T}|\leq 2^Q\). For this fixed \(z\), we have
$$
\begin{aligned}
    \mathbb P(\psi(T,z)=M\mid Z=z)
    &=
    \frac1L
    \sum_{m=1}^{L}
    \mathbb P(\psi(T,z)=m\mid M=m,Z=z) \\
    &=
    \frac1L
    \sum_{m=1}^{L}
    \sum_{\substack{t\in\mathcal{T}:\\ \psi(t,z)=m}}
    \mathbb P(T=t\mid M=m,Z=z) \\
    &=
    \frac1L
    \sum_{t\in\mathcal{T}}
    \mathbb P(T=t\mid M=\psi(t,z),Z=z) \\
    &\leq
    \frac{|\mathcal{T}|}{L}
    \leq
    \frac{2^Q}{L}.
\end{aligned}
$$
Averaging over \(Z\) proves the claim. The key points are that \(M\) and \(Z\) are independent, so conditioning on \(Z=z\) does not change the uniform prior on \(M\), and that the transcript alphabet \(\mathcal{T}\) has size at most \(2^Q\). Any adaptivity in the protocol only changes the conditional law of \(T\) given \((M,Z=z)\); the counting bound uses only the alphabet-size constraint.
Combining the lower and upper bounds on the success probability gives
$$
    1-\epsilon(k')-\delta_{\mathrm{TV}}
    \leq
    2^{Q-k'}.
$$
Taking logarithms yields
$$
    Q
    \geq
    k'
    +
    \log_2(1-\epsilon(k')-\delta_{\mathrm{TV}})
    =
    k'
    -
    \log_2
    \frac{1}{1-\delta_{\mathrm{TV}}-\epsilon(k')}.
$$
Maximizing over $k'$ proves Theorem~\ref{thm:channel-synthesis-converse}.
\end{proof}

\section{Proofs for the capacity calculation}
\label{app:proofs-capacity}
\subsection{A concentration inequality}
\label{app:proof-spherical-tube}

\begin{lemma}[Tube bound around a subspace]
\label{lem:spherical-tube-subspace}
Let \(V\subset\mathbb R^d\) be a fixed \(r\)-dimensional subspace, and let
\(v\sim\operatorname{Unif}(\mathbb S^{d-1})\).  Write \(m=d-r\).  For \(0<\theta<1\),
$$
    \mathbb P\left(
        \operatorname{dist}(v,V)\leq \theta
    \right)
    \leq
    \left(
        C\sqrt{\frac d m}\,\theta
    \right)^m,
$$
where \(C>0\) is universal and \(\operatorname{dist}(v,V)=\|P_{V^\perp}v\|\).  In
particular, if \(r\leq d/2\), this probability is at most \((C\theta)^{d-r}\).  By
rotational invariance, the same bound holds for a fixed unit vector \(v\) and a
uniformly random rank-\(r\) subspace \(V\).
\end{lemma}

\begin{proof}
By rotational invariance, take $V=\operatorname{span}(e_1,\dots,e_r)$. Let $g\sim\mathcal N(0,I_d)$, so $v=g/\|g\|$ is uniform on $\mathbb S^{d-1}$. Then
$$
    \operatorname{dist}(v,V)^2
    =
    \frac{\sum_{i=r+1}^d g_i^2}{\sum_{i=1}^d g_i^2}.
$$
If $A=\sum_{i=r+1}^d g_i^2$ and $B=\sum_{i=1}^r g_i^2$, then $A\sim\chi_m^2$, $B\sim\chi_r^2$, and $A,B$ are independent. Hence
$$
    U:=\operatorname{dist}(v,V)^2=\frac{A}{A+B}
    \sim
    \operatorname{Beta}\left(\frac m2,\frac r2\right).
$$
This beta-distribution identity is the standard projection law for a uniform point on the sphere; see, for example, \cite[Chapter 3]{vershynin2018high}.

Let $a=m/2$, $b=r/2$, and $u=\theta^2$. For $u\leq1/4$,
$$
\begin{aligned}
    \mathbb P(U\leq u)
    &=
    \frac{1}{\mathrm B(a,b)}
    \int_0^u t^{a-1}(1-t)^{b-1}\,dt  \\
    &\leq
    \frac{C}{a\,\mathrm B(a,b)}u^a .
\end{aligned}
$$
By Stirling's formula,
$$
    \frac{1}{a\,\mathrm B(a,b)}
    =
    \frac{\Gamma(a+b)}
         {\Gamma(a+1)\Gamma(b)}
    \leq
    \left(
        C\frac{a+b}{a}
    \right)^a
    =
    \left(
        C\frac d m
    \right)^{m/2}.
$$
Therefore
$$
    \mathbb P(U\leq \theta^2)
    \leq
    \left(
        C\sqrt{\frac d m}\,\theta
    \right)^m.
$$
If $\theta^2>1/4$, the same bound holds after increasing $C$, since the right-hand side is then at least one. This proves Lemma~\ref{lem:spherical-tube-subspace}.
\end{proof}

\subsection{Proof of Theorem~\ref{thm:capacity-random-low-rank}}
\label{app:proof-random-low-rank}

\begin{proof}
We construct an explicit one-shot code. Fix a rank \(r\) to be chosen later. For a
rank-\(r\) subspace \(U\subset\mathbb R^d\), let \(P_U\) be the orthogonal projection
onto \(U\), and define
\begin{equation}
    \Sigma_U
    =
    P_U+\frac1\kappa(I-P_U).
    \label{eq:cap-Sigma-U}
\end{equation}
Then \(\Sigma_U\) has eigenvalue \(1\) on \(U\) and eigenvalue \(1/\kappa\) on
\(U^\perp\), so \(\mathcal N(0,\Sigma_U)\in\mathcal C_d(\kappa)\).

Given input subspace \(U\), the channel output can be written as
\begin{equation}
    Y
    =
    QG+\frac1{\sqrt\kappa}Q_\perp H,
    \label{eq:cap-Y-representation}
\end{equation}
where \(Q\) is an orthonormal basis for \(U\), \(Q_\perp\) is an orthonormal basis for
\(U^\perp\), and \(G\sim\mathcal N(0,I_r)\), \(H\sim\mathcal N(0,I_{d-r})\) are
independent.

For a nonzero vector \(y\), define
\begin{equation}
    \operatorname{dist}(y,U)
    =
    \frac{\|P_{U^\perp}y\|}{\|y\|}.
    \label{eq:cap-dist-def}
\end{equation}
Choose constants \(a_\rho>0\) and \(b_\rho>0\) so that, uniformly over
\(1\leq r\leq d/2\),
\begin{equation}
    \mathbb P(\|G\|<a_\rho\sqrt r)
    +
    \mathbb P(\|H\|>b_\rho\sqrt d)
    \leq
    \frac{\rho}{4}.
    \label{eq:cap-chi-tail}
\end{equation}
Such constants exist by standard chi-square tail bounds. Define the good event
\begin{equation}
    \mathcal E_{\mathrm{good}}
    =
    \{\|G\|\geq a_\rho\sqrt r\}
    \cap
    \{\|H\|\leq b_\rho\sqrt d\}.
    \label{eq:cap-good-event}
\end{equation}
By \eqref{eq:cap-chi-tail},
\[
    \mathbb P(\mathcal E_{\mathrm{good}}^c)
    \leq
    \frac{\rho}{4}.
\]
On \(\mathcal E_{\mathrm{good}}\), \eqref{eq:cap-Y-representation} gives
\[
    \|P_{U^\perp}Y\|
    =
    \kappa^{-1/2}\|H\|,
    \qquad
    \|Y\|
    \geq
    \|P_UY\|
    =
    \|G\|.
\]
Combining this with \eqref{eq:cap-dist-def} and \eqref{eq:cap-good-event}, we get
\[
    \operatorname{dist}(Y,U)
    \leq
    \frac{\kappa^{-1/2}\|H\|}{\|G\|}
    \leq
    \frac{b_\rho}{a_\rho}\sqrt{\frac d{\kappa r}}.
\]
Set
\begin{equation}
    A_\rho=\frac{b_\rho}{a_\rho},
    \qquad
    \theta
    =
    A_\rho\sqrt{\frac d{\kappa r}}.
    \label{eq:cap-theta}
\end{equation}
Thus, on the good event, \(\operatorname{dist}(Y,U)\leq\theta\).

Now draw \(M_{\mathrm{code}}\) independent rank-\(r\) subspaces
\(U_1,\dots,U_{M_{\mathrm{code}}}\) uniformly from the Grassmannian. Message \(m\) is
encoded by the covariance \(\Sigma_{U_m}\) from \eqref{eq:cap-Sigma-U}. The decoder
receives \(Y\) and outputs the index of the nearest codeword subspace:
\begin{equation}
    \widehat m
    \in
    \operatorname*{argmin}_{1\leq \ell\leq M_{\mathrm{code}}}
    \operatorname{dist}(Y,U_\ell).
    \label{eq:cap-decoder}
\end{equation}

Condition on message \(m\) and on \(\mathcal E_{\mathrm{good}}\). By
\eqref{eq:cap-theta}, the true codeword satisfies \(\operatorname{dist}(Y,U_m)\leq
\theta\). Hence an error can occur only if some incorrect \(U_\ell\) satisfies
\(\operatorname{dist}(Y,U_\ell)\leq\theta\). Conditional on \(Y\), each incorrect
\(U_\ell\) is an independent uniform rank-\(r\) subspace. For \(\kappa\) sufficiently
large, the rank choice below ensures \(r\leq d/2\), so
Lemma~\ref{lem:spherical-tube-subspace} gives
\[
    \mathbb P\left(
        \operatorname{dist}(Y,U_\ell)\leq \theta
        \,\middle|\, Y
    \right)
    \leq
    (C\theta)^{d-r}.
\]
Therefore the average error probability of the random code is at most
\[
    \frac{\rho}{4}
    +
    M_{\mathrm{code}}(C\theta)^{d-r}.
\]
Choose
\begin{equation}
    M_{\mathrm{code}}
    =
    \left\lfloor
        \frac{\rho}{4}(C\theta)^{-(d-r)}
    \right\rfloor .
    \label{eq:cap-M-code}
\end{equation}
Then the expected average error is at most \(\rho/2\), and hence there exists a
deterministic code with average error at most \(\rho\).

From \eqref{eq:cap-M-code} and \eqref{eq:cap-theta}, and using
\(\theta=A_\rho\sqrt{d/(\kappa r)}\), we obtain
\begin{align}
    \log M_{\mathrm{code}}
    &\geq
    (d-r)\log\frac1{C\theta}
    +\log\frac{\rho}{8}                                           \notag \\
    &=
    \frac{d-r}{2}
    \left(
        \log\kappa+\log\frac r d
    \right)
    -
    (d-r)\log(CA_\rho)
    +
    \log\frac{\rho}{8}. 
    \label{eq:cap-message-bits-expanded}
\end{align}
Here \(C>0\) is the universal constant from Lemma~\ref{lem:spherical-tube-subspace},
whereas \(A_\rho=b_\rho/a_\rho\) depends only on \(\rho\).  Since \(d\geq2\) and
\(d-r\leq d\), the last two terms in \eqref{eq:cap-message-bits-expanded} can be
absorbed into \(-C_\rho d\), where \(C_\rho>0\) depends only on \(\rho\).  Hence
\begin{equation}
    \log M_{\mathrm{code}}
    \geq
    \frac{d-r}{2}
    \left(
        \log\kappa+\log\frac r d
    \right)
    -
    C_\rho d 
    \label{eq:cap-message-bits-lb}
\end{equation}
for some $C_\rho>0$ that may depend on $\rho$.
Now choose
\begin{equation}
    r
    =
    \max\left\{
        1,
        \left\lfloor
            \frac d{\log\kappa}
        \right\rfloor
    \right\}.
    \label{eq:cap-r-choice}
\end{equation}
The factor $\frac1{\log\kappa}$ in \eqref{eq:cap-r-choice} arises from optimizing $r$ in \eqref{eq:cap-message-bits-lb},
but is not critical for our final conclusion.
For \(\kappa\) sufficiently large, this choice satisfies \(r\leq d/2\). If
\(\log\kappa\leq d/2\), then \(r\asymp d/\log\kappa\), and
\eqref{eq:cap-message-bits-lb} implies
\[
    \log M_{\mathrm{code}}
    \geq
    \frac d2\log\kappa
    -
    \frac d2\log\log\kappa
    -
    C_\rho d
    \geq
    c_\rho d\log\kappa
\]
for sufficiently large \(\kappa\). If \(\log\kappa>d/2\), then \(r=1\), and
\eqref{eq:cap-message-bits-lb} gives
\[
    \log M_{\mathrm{code}}
    \geq
    \frac{d-1}{2}
    \left(
        \log\kappa-\log d
    \right)
    -
    C_\rho d
    \geq
    c_\rho d\log\kappa
\]
again for sufficiently large \(\kappa\). Hence
\[
    \epsilon_{\mathrm G}(k')\leq \rho
    \qquad
    \text{for every }
    k'\leq c_\rho d\log\kappa .
\]

Finally, the capacity lower bound follows from the same code. Let \(M\) be uniform over
the codebook messages. Since the average decoding error is at most \(\rho\), Fano's
inequality gives
\[
    I(M;Y)
    \geq
    (1-\rho)\log M_{\mathrm{code}}-1.
\]
The induced prior over covariance matrices is admissible, so
\[
    C_{2,d}(\kappa)
    \geq
    I(M;Y)
    \geq
    c\,d\log\kappa
\]
after adjusting constants. This proves Theorem~\ref{thm:capacity-random-low-rank}.
\end{proof}

\subsection{Proof of Theorem~\ref{thm:fixed-kappa-large-d}}
\label{app:proof-fixed-kappa-large-d}

We prove Theorem~\ref{thm:fixed-kappa-large-d} by restricting the Gaussian covariance class to a diagonal binary subchannel. Fix \(\kappa>1\). Consider the scalar binary-input channel \(W_\kappa\) with input \(B\in\{0,1\}\) and output \(Y\in\mathbb R\) given by
\[
    B=0:\quad Y\sim \mathcal N(0,1/\kappa),
    \qquad
    B=1:\quad Y\sim \mathcal N(0,1).
\]
This channel has strictly positive capacity because the two output distributions are distinct. Let
\[
    C_\kappa
    =
    \max_{P_B} I(B;Y)
    >
    0.
\]
Choose any rate \(R_\kappa\in(0,C_\kappa)\). By the standard random-coding theorem, equivalently the positivity of the random-coding error exponent below capacity, there exists \(E_\kappa>0\) such that, for all sufficiently large \(d\), the \(d\)-fold product channel \(W_\kappa^{\otimes d}\) admits a code with
\[
    M_d
    =
    \left\lfloor 2^{R_\kappa d}\right\rfloor
\]
messages and average decoding error at most
\[
    2^{-E_\kappa d}.
\]
See, for example, the classical channel-coding error exponent bound of Gallager
\citep{gallager1968information}.

We now embed this product channel into the Gaussian covariance class
\[
    \mathcal C_d(\kappa)
    =
    \left\{
        \mathcal N(0,\Sigma):
        \frac1\kappa I\preceq \Sigma\preceq I
    \right\}.
\]
For each binary vector \(b=(b_1,\dots,b_d)\in\{0,1\}^d\), define the diagonal covariance
\[
    \Sigma_b
    =
    \operatorname{diag}(\sigma_1^2,\dots,\sigma_d^2),
    \qquad
    \sigma_i^2
    =
    \begin{cases}
        1/\kappa, & b_i=0,\\
        1, & b_i=1.
    \end{cases}
\]
Then \(\mathcal N(0,\Sigma_b)\in\mathcal C_d(\kappa)\), and drawing
\(Y\sim\mathcal N(0,\Sigma_b)\) is exactly one use of the product channel
\(W_\kappa^{\otimes d}\) with input \(b\).

Let the allowed sampling error satisfy
\[
    1-\delta_{\mathrm{TV}}=2^{-ad},
\]
where \(a>0\) will be chosen small enough. Set
\[
    \alpha_d
    =
    \frac{1-\delta_{\mathrm{TV}}}{2}
    =
    2^{-ad-1}.
\]
If \(a\leq E_\kappa/2\), then for all sufficiently large \(d\),
\[
    2^{-E_\kappa d}
    \leq
    2^{-ad-1}
    =
    \alpha_d.
\]
Thus the embedded Gaussian covariance channel admits a one-shot code with
\[
    k'_d
    =
    \left\lfloor R_\kappa d\right\rfloor
\]
message bits and decoding error at most \(\alpha_d\).

Now apply Corollary~\ref{cor:channel-synthesis-fixed-error}, or equivalently the inverse form of Theorem~\ref{thm:channel-synthesis-converse}, with
\[
    \rho=\alpha_d=\frac{1-\delta_{\mathrm{TV}}}{2}.
\]
Any \(\delta_{\mathrm{TV}}\)-accurate finite-bit sampler for
\(\mathcal C_d(\kappa)\) must satisfy
\[
    Q
    \geq
    k'_d
    -
    \log_2\frac{2}{1-\delta_{\mathrm{TV}}}.
\]
Since \(1-\delta_{\mathrm{TV}}=2^{-ad}\), the logarithmic penalty is
\[
    \log_2\frac{2}{1-\delta_{\mathrm{TV}}}
    =
    ad+1.
\]
Therefore
\[
    Q
    \geq
    \left\lfloor R_\kappa d\right\rfloor
    -
    ad
    -
    1.
\]
Choose
\[
    a_\kappa
    =
    \frac12\min\{R_\kappa,E_\kappa\}.
\]
Then for every \(0<a\leq a_\kappa\),
\[
    Q
    \geq
    \frac{R_\kappa}{2}d-2
\]
for all sufficiently large \(d\). Hence, after increasing the threshold on \(d\) if necessary,
\[
    Q
    \geq
    c_\kappa d
\]
with, for instance, \(c_\kappa=R_\kappa/4>0\). This proves Theorem~\ref{thm:fixed-kappa-large-d}.

\paragraph{AI disclosure.}
GPT~5.5 was used in proving the lemmas involved in the upper-bound proofs. GPT~6 was used in editing the final version of the draft.

\bibliographystyle{apalike}
\bibliography{ref}
\end{document}